\documentclass[11pt]{article}
\usepackage[utf8]{inputenc}
\usepackage{geometry}
\usepackage{xcolor}
\usepackage{amsmath}
\usepackage{bm}
\usepackage{bbm}
\usepackage{hyperref}
\usepackage{amsfonts}
\usepackage{array}
\usepackage{natbib}
\usepackage{algorithm}          
\usepackage{enumitem}  
\usepackage{graphicx}
\usepackage{algpseudocode}

\newcommand{\indep}{\perp\!\!\!\perp}
\newtheorem{assumption}{Assumption}
\newtheorem{proposition}{Proposition}

\title{Double Variable Importance Matching to Estimate Distinct Causal Effects on Event Probability and Timing}

\author{
\begin{tabular}{c c c}
Yuqi Li$^{1}$ & Quinn Lanners$^{1}$ & Matthew M. Engelhard$^{1}$ \\
\end{tabular} \\[0.8em]
$^{1}$Department of Biostatistics \& Bioinformatics, Duke University \\[0.3em]
}

\date{}

\begin{document}
\maketitle

\begin{abstract}
In many clinical contexts, estimating effects of treatment in time-to-event data is complicated not only by confounding, censoring, and heterogeneity, but also by the presence of a cured subpopulation in which the event of interest never occurs. In such settings, treatment may have distinct effects on (1) the probability of being cured and (2) the event timing among non-cured individuals. Standard survival analysis and causal inference methods typically do not separate cured from non-cured individuals, obscuring distinct treatment mechanisms on cure probability and event timing. To address these challenges, we propose a matching-based framework that constructs distinct match groups to estimate heterogeneous treatment effects (HTE) on cure probability and event timing, respectively. We use mixture cure models to identify feature importance for both estimands, which in turn informs weighted distance metrics for matching in high-dimensional spaces. Within matched groups, Kaplan–Meier estimators provide estimates of cure probability and expected time to event, from which individual-level treatment effects are derived. We provide theoretical guarantees for estimator consistency and distance metric optimality under an equal-scale constraint. We further decompose estimation error into contributions from censoring, model fitting, and irreducible noise. Simulations and real-world data analyses demonstrate that our approach delivers interpretable and robust HTE estimates in time-to-event settings. 
\end{abstract}

\section{Introduction}

The effects of treatment on clinical outcomes can often be divided into short-term effects and long-term effects. In maternal health, prescribing prophylactic azithromycin after labor and delivery might delay infection but ultimately fail to prevent it \citep{tita2016adjunctive}. In mental health, psychotherapy might reduce rates of mental health conditions, yet paradoxically hasten recognition of conditions already present\citep{jorm2017has}. 
In early-stage melanoma, immunotherapy might lead to long-term remission (i.e., cure) for some patients but simply delay recurrence in others \citep{michielin2020evolving, boutros2023cured}.
In these and many more cases, we may wish to distinguish between effects of treatment on (a) the long-term probability of an event of interest, and (b) the short-term timing of that event, should it occur. Drawing a clear distinction between these quantities is key, for example when presenting prognostic information to patients to inform treatment decisions, or when exploring underlying treatment mechanisms. 

In the survival analysis literature, the distinction between short-term and long-term (or even indefinite) effects can be modeled by positing the existence of a \textit{cured subpopulation} in which the event of interest never occurs\citep{farewell1982use}.

We may then differentiate between effects of treatment on (a) the probability of being \textit{cured} (\textit{i.e.}, long-term effect), and (b) the time to event distribution among non-cured patients (\textit{i.e.}, short-term effects). As illustrated conceptually in Figure \ref{fig:intro}, treatments can and often do confer long-term benefit even while decreasing the expected time to event should the event take place.
However, the traditional survival analysis framework and associated causal methods typically conflate these two effects, treating all censored individuals as if they might eventually experience the event. Neglecting this distinction yields an imprecise understanding of the effects of treatment, in turn compromising our ability to understand treatment mechanisms and select optimal clinical actions.

\begin{figure}[htbp]
    \centering
    \includegraphics[width=0.65\columnwidth]{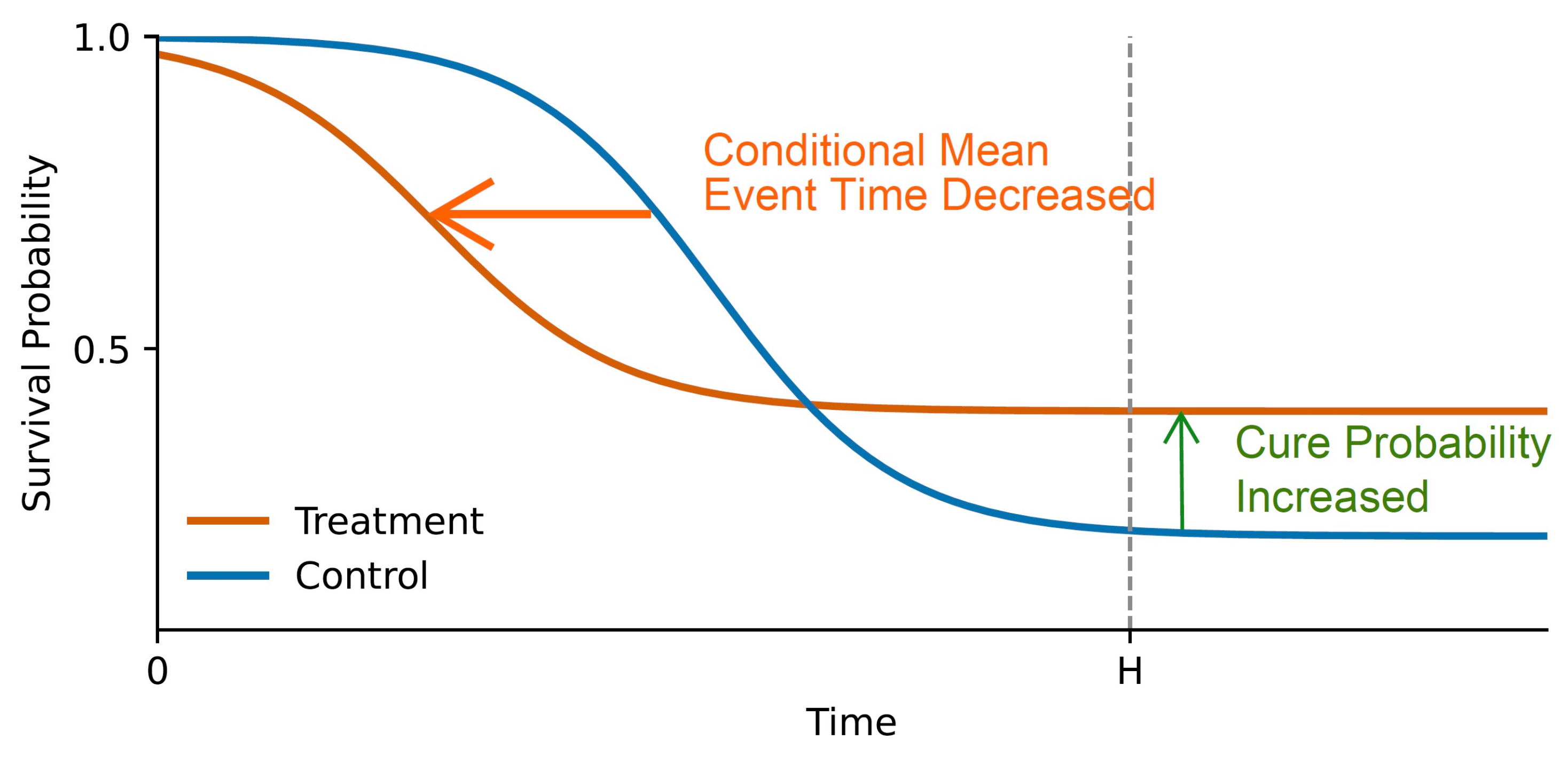}  
    \caption{Hypothetical Survival Curves Where Treatment Increases the Cure Probability yet Reduces the Conditional Mean Event Time.}
    \label{fig:intro}
\end{figure}

In this work, we aim to estimate heterogeneous treatment effects (HTE) on these two clinically distinct components of time-to-event outcomes. To formalize these effects, we define a clinically motivated time horizon $H$ during which relevant events may occur. Patients who do not experience an event before $H$ are considered \textit{cured} (\textit{i.e.,} event-free). Events occurring after $H$ can be disregarded, as they are outside of the finite window relevant to the original condition and its associated treatment.
Under this framework, we define treatment effect estimands for (i) the probability of being cured and (ii) the expected event time among those who do experience an event.

We then propose a matching-based approach for estimating the two types of HTEs. Matching methods are widely used in causal inference for their conceptual simplicity and transparency, as they closely mimic the design of a randomized study \citep{stuart2010matching}. These qualities are particularly valuable in high-stakes applications such as clinical research, where time-to-event data are common.  In our setting, we construct separate matched groups for the two estimands, recognizing that event probability and timing may depend on different covariates. A key challenge is that observational data is often high-dimensional and contains many irrelevant or noisy features. Recent work addresses this issue while preserving the interpretability by using auxiliary outcome models to learn distance metrics that upweight the covariates most associated with the outcome \citep{lanners2023variable, parikh2022malts}. We adapt this approach to our setting by learning two distinct distance metrics  to match on covariates most relevant to cure status and to timing, respectively. We fit a mixture cure model \citep{farewell1982use} whose likelihood jointly incorporates information about cure mechanism and the event time distribution. Then we extract separate feature importances specific to each effect of interest and use them to construct two tailored distance metrics for matching.

Using these distance metrics, we implement a double matching procedure to estimate HTE. For each individual, we construct two separate matched groups based on the two distinct distance metrics learned from the mixture cure model. Each distance metric approximates the optimal covariate weighting for its respective estimand, improving match quality and reducing estimation error. Within each treatment arm, we apply Kaplan-Meier estimators to estimate the survival probability at time $H$, interpreted as cure probability, and the mean event time conditional on event occurrence before $H$. The difference between matched treated and control groups yields interpretable estimates of HTE. We also show that, under standard assumptions, the proposed estimators are consistent for both causal effects. 

Our contributions are to (1) develop a novel matching framework that separately estimates HTE on cure probability and event timing; (2) construct cure- and timing-specific distance metrics for matching based on covariate importances from an auxiliary mixture cure model; (3) show these metrics are optimal for HTE estimation under an equal-scale constraint; (4) demonstrate that Kaplan-Meier estimators within matched groups yield consistent estimates of individual-specific survival curves under standard assumptions; (5) validate our method across diverse simulation settings and a real-world case study; and (6) decompose HTE estimation error into components attributable to censoring, model fitting, and irreducible noise, offering insight into sources and limits of estimation uncertainty.

\section{Literature Review}
With the growing availability of large-scale datasets, there has been increased focus on estimating HTE \citep{wager2018estimation, knaus2021machine}, for which there are numerous regression, weighting, and machine learning methods \citep{yao2021survey}. However, because of the reliance on untestable assumptions, recent work has emphasized interpretable approaches that can be scrutinized by domain experts \citep{shikalgar2025two,parikh2023effects}. Matching methods are especially appealing in this regard, as their conceptual simplicity facilitates incorporating expert judgment into the analysis pipeline \citep{stuart2010matching, Ruoqi2021}

While propensity \citep{rosenbaum1983central} and prognostic score matching \citep{hansen2008prognostic} remain the most widely used approaches, a growing class of almost-matching-exactly (AME) frameworks has been developed to improve interpretability and accuracy in treatment effect estimation \citep{yu2021optimal, diamond2013genetic, wang2021flame, dieng2019interpretable, morucci2020adaptive, katta2024interpretable}. Closest to the approach we employ are Matching After Learning to Stretch (MALTS) \citep{parikh2022malts} and Variable Importance Matching (VIM), which use a stretched distance metric that upweights important variables. Our method extends VIM to time-to-event settings with an event-free subpopulation.

\vspace{-1em} 
\paragraph{Matching for Causal Inference with Time-to-Event Data}

Matching methods are especially appealing in survival settings, where censoring and variable follow-up complicate estimation. In this context, treatment effects can be defined on several estimands, such as survival curves, hazard ratios, or scalar summaries like restricted mean or median survival times \citep{bewick2004statistics}. Extensions of matching have been developed to target each of these quantities: survival curves \citep{austin2014use, xu2023causal}, hazard ratios \citep{austin2015optimal, lee2020empirical}, and restricted mean or median survival times \citep{lin2023matched, yang2019survival, jiang2024causal}. These approaches include widely used frameworks such as propensity score matching \citep{austin2014use, lu2018testing, xu2023causal, lin2023matched}, as well as more recent innovations like sufficient kernel matching \citep{jiang2024causal}, optimal full matching \citep{austin2015optimal, lee2020empirical}, and coarsened exact matching \citep{yang2019survival}.

Most matching approaches for survival data focus on population-level estimands. To estimate conditional treatment effects, modeling-based methods such as Cox regression, AFT BART, and targeted maximum likelihood estimation have been developed \citep{ashouri2024generalized, zhu2020targeted, austin2014use}. Although these methods allow for conditional treatment effect estimation, they are vulnerable to model misspecification, and they are less transparent and thus harder to audit than matching-based methods, limiting trustworthiness in high-stakes settings.

\vspace{-1em} 
\paragraph{Survival Analysis
with a Cured Subpopulation}

Mixture cure and promotion time cure models address time-to-event data with a cured subpopulation, differing in how they represent the cured group \citep{amico2018cure, boag1949maximum, berkson1952survival, yakovlev1996stochastic, chen1999new, tsodikov1998proportional}. These models have been studied for decades, and more recent work has relaxed parametric assumptions through semiparametric and machine learning approaches \citep{farewell1982use, felizzi2021mixture, engelhard2022disentangling, musta20222}. We focus on the mixture cure framework, as it aligns more directly with our estimands of interest.

There is limited existing work on causal inference using time-to-event data with a cured subpopulation. \cite{gao2017estimating} used a promotion time cure model to estimate treatment effects in randomized controlled trials with noncompliance. Most similar to our objective are the mixture cure approaches of \cite{sun2025tree} and \cite{wang2024causal}. \cite{sun2025tree} propose an adaptation of Bayesian causal forest combined with AFT mixture cure models whereas  \cite{wang2024causal} propose semi-parametric estimators for this task. Although effective, these modeling-based approaches yield results that are less transparent and harder to communicate to domain experts than matching-based methods. Matching-based approaches that explicitly address a cured subpopulation and jointly target cure probability and survival among the uncured remain underexplored.

\section{Setup and Estimands}
Consider a dataset $\mathcal{D}$ consisting of $n$ subjects, labeled $i = 1, \ldots, n$. For each subject $i$, let $Z_i$ denote the treatment status, where $Z_i = 1$ indicates active treatment and $Z_i = 0$ indicates control. Each subject has an event time $T_i$, subject to right censoring by $C_i$. We also observe a vector of $d$ pre-treatment covariates, $\bm{X}_i \in \mathcal{X} \subseteq \mathbb{R}^d$, drawn i.i.d.\ across subjects.

To define causal estimands, we adopt the potential outcomes framework \citep{Rubin1980}, under which each subject has two potential failure times $(T_i(0), T_i(1))$ and two potential censoring times $(C_i(0), C_i(1))$, one under each treatment. Since only one potential outcome is observed for each subject, the individual-level treatment effects are not identifiable. Instead, we focus on heterogeneous treatment effects (HTEs) on these time-to-event outcomes.

We define the potential survival probability function under treatment $z$ as $S_{z}(t | \bm{x}) = \Pr(T(z) > t | \bm{X} = \bm{x})$,  $\forall \bm{x}\in\mathcal{X}$. A common causal estimand is the conditional HTE on survival at time $t$:
\begin{equation}
    \tau(t | \bm{x}) = S_{1}(t | \bm{x}) - S_{0}(t | \bm{x}),\quad 0 \leq t \leq t_{\max},
\end{equation}
where $t_{\max}$ is the maximum follow-up time. This estimand is easily intepreted as the causal risk difference between the two arms. However, since it is a function of the time $t$, estimation of $\tau(t|\bm{x})$ is challenging, especially for large $t$ values where few events occur. 

To overcome these difficulties, we set a prespecified time threshold $H$ and focus on the event time only of the events that occur prior to $H$. This approach aligns with common practice in many clinical studies because (a) pharmacologic and other treatments are active only for a limited time, and (b) indefinite follow-up is not practical. Due to potential censoring and the imposed time horizon $H$, we observe only the event time $Y_i = \min\{T_i, C_i, H\}$ and the event indicator $\delta_i = \mathbb{I}(T_i < C_i, T_i < H)$.  

With this setup, we define a latent variable $E_i(z) = \mathbb{I}(T_i(z) < H)$, where $E_i(z) = 0$ indicates the individual is \textit{cured} and $E_i(z) = 1$ indicates the individual is \textit{non-cured} (both under treatment $z$). This allows us to define the following two estimands of interest: (i) the HTE on the \textbf{cure probability}:
\begin{equation}
    \pi(\bm{x}) = \Pr(E_i(1)=0 \mid \bm{X} = \bm{x}) - \Pr(E_i(0)=0 \mid \bm{X} = \bm{x}),
\end{equation}
and (ii) the HTE on the \textbf{conditional mean event time (CMET)},
\begin{equation}
    \begin{aligned}
    \Delta(\bm{x}) &= \mathbb{E}[T_i(1) \mid T_i(1) < H, \bm{X} = \bm{x}] \\
        &- \mathbb{E}[T_i(0) \mid T_i(0) < H, \bm{X} = \bm{x}],        
    \end{aligned}
\end{equation}

\vspace{-1em}
which is conditioning on the \textit{non-cured} unit. We concentrate on estimating $\pi(\bm{x})$ and $\Delta(\bm{x})$ hereafter.

\subsection{Assumptions and Identification}\label{sec:identification}

We adopt Stable Unit Treatment Value Assumption (SUTVA) \citep{Rubin1980} which consists of consistency and no interference assumptions. In addition, we assume the following standard conditions for identification: unconfoundedness, positivity and non-informative censoring \citep{ImbensRubin2015}. Under these assumptions, the following proposition shows that both $\pi(x)$ and $\Delta(x)$ are identifiable with conditional survival probability function $S_z(t \mid x) = \Pr(T > t \mid Z = z, \bm{X} = \bm{x})$. The details of assumptions and derivations are in the supplementary materials.

\begin{proposition}\label{prop:identification}
    If the consistency, unconfoundedness, positivity and non-informative censoring assumptions hold, then the two estimands can be expressed as
    \begin{equation}\label{eq:cure-estimand}
        \pi(\bm{x})  = S_1(H \mid \bm{x}) - S_0(H \mid \bm{x}),
    \end{equation}
    \begin{equation}\label{eq:cmst-estimand}
        \begin{aligned}
            \Delta(\bm{x}) &= \frac{\int_0^H S_1(t|\bm{x}) \, dt -H S_1(H|\bm{x})}{1 - S_1(H|\bm{x})} \\
                      &- \frac{\int_0^H S_0(t|\bm{x}) \, dt -H S_0(H|\bm{x})}{1 - S_0(H|\bm{x})}.
        \end{aligned}        
    \end{equation}
\end{proposition}

Proposition \ref{prop:identification} indicates that the key to the estimation of $\pi(\bm{x})$ and $\Delta(\bm{x})$ lies in the survival probability functions $S_z(t \mid \bm{x})$ for $z \in \{0, 1\}$ and $t \in [0, H]$. This characterization enables nonparametric estimation of the target estimands via Kaplan–Meier estimators applied to matched groups, as we describe in Section~\ref{sec:method-KM}.


\subsection{Mixture Cure Model}
A mixture cure model is comprised of two components that model (a) the cure probability, and (b) the event time distribution. Once the forms of both models are specified, we can estimate model parameters via maximum likelihood estimation. Given that we can only observe time $Y$ and the event indicator $\delta$, then in the mixture cure model, we consider the density $p(y,\delta|\bm{x})$, which is a mixture of cure status and event time information corresponding to non-cured versus cured individuals, respectively. The density and survival function can be expressed by the two components of the mixture cure model, the detailed derivations are provided in the supplementary materials.

With a pre-specified time horizon $H$, the survival function modeled by a mixture cure approach can be written as: 
\begin{equation}
    \begin{aligned}
        S(t) &= \Pr(cure) + [1-\Pr(cure)]S_N(t) \\
        &= \Pr(E=0) + [1 - \Pr(E=0)] \Pr( T>t | T \leq H) \\
        &= \Pr(T>H) + \Pr(T \leq H) \Pr( T>t | T \leq H) \\
        &= \Pr(T>t)
    \end{aligned}
\end{equation}
where $\Pr(cure)$ is the probability of being cured, and $S_N(t)$ is the survival function given the patient is non-cured.  

The density function can be summarized as follows:

\begin{equation}\label{eq: density}
    \begin{aligned}
        p(y, &\delta | \bm{x}) 
        = p(y, \delta=1|\bm{x}) + p(y, \delta=0|\bm{x}) \\
        &= p(E=1|\bm{x}) \mathbb{I}(y<H) \frac{f_T(y|\bm{x})}{F_T(H|\bm{x})} S_C(y|\bm{x}) \\
        &+ p(E=1|\bm{x}) \mathbb{I}(y<H) \left( 1 - \frac{F_T(y|\bm{x})}{F_T(H|\bm{x})} \right) f_C(y|\bm{x}) \\
        &+ p(E=0|\bm{x}) \mathbb{I}(y<H) f_C(y|\bm{x}) \\
        &+ p(E=0|\bm{x}) \mathbb{I}(y=H) S_C(y|\bm{x}) 
    \end{aligned}    
\end{equation}

where $f_T(y|\bm{x})$ and $F_T(H|\bm{x})$ are the \textit{pdf} and \textit{cdf} of the event time distribution, and $f_C(y|\bm{x})$ and $S_C(y|\bm{x})$ are the \textit{pdf} and \textit{1 - cdf} of the censoring time distribution. Under the non-informative censoring assumption, the censoring terms are nuisance components.

The parametric models are often straightforward to fit and yield closed-form $\pi(\bm{x})$ and $\Delta(\bm{x})$, but heavily rely on correct model specification. In contrast, the Kaplan-Meier estimator is a non-parametric approach that makes minimal assumptions but cannot adjust for covariates. Therefore, we propose a framework combining both parametric and non-parametric approaches, which leads to smaller estimation error and is robust to model misspecification.

\section{Methodology}\label{sec:methodology}

We present a matching-based approach for causal inference with time-to-event data that accounts for a cured subpopulation. 
Like similar matching approaches \citep{parikh2022malts,lanners2023variable}, our method learns a distance metric to form matched groups for treatment effect estimation. We extend this framework by using mixture cure models to derive two distinct metrics, one for cure probability and one for CMET, enabling targeted matching for each estimand. Our approach is summarized in Algorithm~\ref{alg:mcm_matching}, and the following three subsections describe each step in detail.

\subsection{Fitting Mixture Cure Models}
To facilitate honest causal inference, we begin by splitting the dataset into a training ($\mathcal{D}_{\text{tr}}$) and estimation set ($\mathcal{D}_{\text{est}}$) \citep{ratkovic2019rehabilitating}. Using only the samples in $\mathcal{D}_{\text{tr}}$,
for each treatment arm $Z\in\{0,1\}$, we fit a mixture cure model with two parts:
\vspace{-0.5em}
\begin{align}\label{eq:cure-prob-mcm}
    P(\text{cure}|\mathbf{X}) &= \sigma(\beta_{0} + \mathbf{X}\boldsymbol{\beta}), 
    \quad \sigma(z)=\tfrac{1}{1+e^{-z}},    
\end{align}
\vspace{-1.7em}
\begin{align}\label{eq:aft-mcm}
    \log(T) &\sim N(\lambda_{0} + \mathbf{X}\boldsymbol{\lambda}, \sigma^2),
\end{align}
where Equation~\ref{eq:cure-prob-mcm} specifies cure probability and Equation~\ref{eq:aft-mcm} specifies an AFT model for event time distribution. Model parameters are estimated by maximizing the likelihood $\mathcal{L} = \sum_{i=1}^N \log p(y_i, \delta_i|\bm{x_i})$ on the training set, and the resulting coefficients $\boldsymbol\beta,\boldsymbol\lambda$ summarize variable importance for the two outcome components.

\subsection{Learning Distance Metrics}
We use the estimated coefficients to quantify feature importance and construct weighted Euclidean distances:
\vspace{-1em}
\begin{equation}\label{eq: distance metric}
\begin{aligned}
d_{\text{cure}}(\mathbf{x},\mathbf{y}) &= \sqrt{ (\mathbf{x}-\mathbf{y})^\top \mathbf{W}_{\text{cure}} (\mathbf{x}-\mathbf{y}) }, \\
d_{\text{time}}(\mathbf{x},\mathbf{y}) &= \sqrt{ (\mathbf{x}-\mathbf{y})^\top \mathbf{W}_{\text{time}} (\mathbf{x}-\mathbf{y}) }, \\
\mathbf{W}_{\text{cure}} &= \frac{1}{2} \, \mathrm{diag}\big(|\boldsymbol{\beta}_1| + |\boldsymbol{\beta}_0|\big), \\
\mathbf{W}_{\text{time}} &= \frac{1}{2} \, \mathrm{diag}\big(|\boldsymbol{\lambda}_1| + |\boldsymbol{\lambda}_0|\big),
\end{aligned}
\end{equation}

\vspace{-0.5em}

where $|\boldsymbol{\beta}_z|$ and $|\boldsymbol{\lambda}_z|$ are vectors of absolute values of the estimated coefficients from the cure and time parts of the mixture cure model in treatment arm $Z=z$, and $\mathrm{diag}(\cdot)$ denotes a diagonal matrix. These weight matrices assign larger importance to covariates more strongly associated with the corresponding outcomes, thereby approximating exact matching on the most influential features.


\subsection{Matching and HTE Estimation} \label{sec:method-KM}
For each unit $i$ in the estimation set, $\mathcal{D}_{\text{est}}$, we first identify matched sets by finding $K$ nearest neighbors from the treated ($M_1$) and control ($M_0$) cohorts using the cure and time distance metrics separately. Based on these matched sets, we then estimate HTE using Kaplan--Meier estimators. 

According to Section \ref{sec:identification}, the HTE on cure probability is estimated using matched groups identified by the cure distance $d_{\text{cure}}$:
\begin{align}\label{eq:hte_cure_estimator}
    \pi(\mathbf{x}_i) = \hat S_{M_1}(H) - \hat S_{M_0}(H),
\end{align}
where $\hat S$ denotes the Kaplan--Meier estimator.

Similarly, the HTE on CMET is estimated using matched groups identified by the time distance $d_{\text{time}}$:
\begin{align}\label{eq:hte_time_estimator}
    \Delta(\bm{x}_i) &= \frac{\int_0^H \hat S_{M_1}(t) \, dt -H \hat S_{M_1}(H)}{1 - \hat S_{M_1}(H)} \notag\\ 
    &- \frac{\int_0^H \hat S_{M_0}(t) \, dt -H \hat S_{M_0}(H)}{1 - \hat S_{M_0}(H)}.
\end{align}

\begin{algorithm}[htbp]
\caption{Double Variable Importance Matching for HTE Estimation}
\label{alg:mcm_matching}
\begin{algorithmic}[1]
\Require Data $\mathcal{D} = \{\mathbf{X}_i, Y_i, \delta_i, Z_i\}_{i=1}^N$, time horz. $H$
\Ensure HTE estimates for cure and survival time
\State \textbf{Split Data}: Divide $\mathcal{D}$ into training set $\mathcal{D}_{\text{tr}}$ (35\%) and estimation set $\mathcal{D}_{\text{est}}$ (65\%)
\State \textbf{Train Models}:
\begin{itemize}
    \item Fit mixture cure models on $\mathcal{D}_{\text{tr}}(Z=0)$ and $\mathcal{D}_{\text{tr}}(Z=1)$ respectively
    \item Extract coefficients $\boldsymbol{\beta}_1, \boldsymbol{\beta}_0$ (cure model) and $\boldsymbol{\lambda}_1, \boldsymbol{\lambda}_0$ (time model)
\end{itemize}
\State \textbf{Construct Weight Matrices}:
\begin{itemize}
    \item $\mathbf{W}_{\text{cure}} \gets \tfrac{1}{2} \mathrm{diag}\big(|\boldsymbol{\beta}_1| + |\boldsymbol{\beta}_0|\big)$
    \item $\mathbf{W}_{\text{time}} \gets \tfrac{1}{2} \mathrm{diag}\big(|\boldsymbol{\lambda}_1| + |\boldsymbol{\lambda}_0|\big)$
\end{itemize}
\State \textbf{Match and Estimate}:
\begin{itemize}
    \item For each $i \in \mathcal{D}_{\text{est}}$:
    \begin{enumerate}
        \item Construct matched groups $M_1, M_0$ using $d_{\text{cure}}$ and $d_{\text{time}}$ to estimate cure and time effects, respectively.
        \item Compute HTE on cure probability and conditional mean event time using Kaplan–Meier estimates from matched groups, plugged into the estimators in Equations~\ref{eq:hte_cure_estimator} and~\ref{eq:hte_time_estimator}.
    \end{enumerate}
\end{itemize}
\State \textbf{Return}: HTE estimates for all estimation units.
\end{algorithmic}
\end{algorithm}

\section{Theoretical Results}

In this section, we establish theoretical guarantees for our proposed framework. First, we show that the matching estimators for both cure probability and event timing are consistent under standard assumptions. We further characterize the optimal distance metric for matching under a linear model assumption. Proofs are provided in the supplementary materials. 

\subsection{Consistency of Matching Estimation}  
We first establish the consistency of the proposed matching estimators for cure and time effects. Recall the form of our estimands in Equations~\ref{eq:cure-estimand} and \ref{eq:cmst-estimand}.
We will need to show the consistency of $\hat S(H|\bm{x})$ and $\int_0^H \hat S(t|\bm{x}) \, \mathrm{d}t$. Let $\{\bm{X}_{i_1}, ..., \bm{X}_{i_K}\} \subset \{\boldsymbol{X}_1, \dots, \boldsymbol{X}_n\}$ denote a matching group of size $K$, chosen as the $K$ nearest neighbors of $\bm{x}$ under metric $d(\cdot,\cdot)$. Let $\widehat S_{Kn}$ be the Kaplan–Meier estimator based on $\{(Y_{i_k},\delta_{i_k})\}_{k=1}^K$.

Define two \textbf{functionals} $F(S) := S(H)$ and $G(S) := \int_0^H S(u)\,\mathrm{d}u$. Consistency of the estimators reduces to showing $F(\widehat S_{Kn})\to F(S(\cdot\mid \boldsymbol{x}))$ and $G(\widehat S_{Kn})\to G(S(\cdot\mid \boldsymbol{x}))$.

\textbf{Theorem 1.}
    Let $\boldsymbol{x}$ be in the support of $\boldsymbol{X}$, \emph{i.e.}, for any $\epsilon > 0$, $\Pr(d(\boldsymbol{X}, \boldsymbol{x}) < \epsilon) > 0$. Assume $S(t\mid \boldsymbol{x})$ is locally Lipschitz in $\boldsymbol{x}$ uniformly in $t$, \emph{i.e.}, there exists a neighborhood of $\boldsymbol{x}$, denoted by $\mathcal{N}_{\boldsymbol{x}}$, and a constant $\tau > H$ such that for all $\boldsymbol{x}', \boldsymbol{x}''$ in $\mathcal{N}_{\boldsymbol{x}}$ and $t \in [0, \tau]$, $|S(t \mid \boldsymbol{x}') - S(t \mid \boldsymbol{x}'')| \leq L \cdot d(\boldsymbol{x}', \boldsymbol{x}'')$ for some finite constant $L$. Then, 
    \vspace{-0.2em}
    $$F(\widehat{S}_{Kn}) \stackrel{p}{\rightarrow} F(S(\cdot \mid \boldsymbol{x})), \quad G(\widehat{S}_{Kn}) \stackrel{p}{\rightarrow} G(S(\cdot \mid \boldsymbol{x})),$$
    \vspace{-0.3em}
    as $K, n \rightarrow \infty$ with $K = o(n)$, where $\stackrel{p}{\rightarrow}$ denotes convergence in probability.

Since the consistency of $F(\widehat S_{Kn})$ and $G(\widehat S_{Kn})$ is established, 
the following corollary follows immediately.

\textbf{Corollary 1.}
    The two quantities \ref{eq:hte_cure_estimator} and \ref{eq:hte_time_estimator} defined in the estimators, which are functions of $F(\widehat S_{Kn})$ and $G(\widehat S_{Kn})$, are also consistent.

\subsection{Optimal Distance Metric}
The distance metrics used for matching should adaptively weight features by their influence on the estimands of interest, so that the worst-case discrepancy under local perturbations is minimized. This is formalized by bounding the Lipschitz constant under a weighted Euclidean distance. Assuming a linear mixture cure model, we establish optimal weights for matching-based estimation of cure and time effects.

\paragraph{Setup.} 
For a diagonal weight matrix $W=\mathrm{diag}(w_1,\ldots,w_p)$ with $w_j\ge 0$ and $\sum_{j=1}^p w_j=c>0$, define the weighted norm $\|\boldsymbol{v}\|_W=(\sum_{j=1}^p w_j v_j^2)^{1/2}$ and distance $d_W(\boldsymbol{x},\boldsymbol{x}')=\|\boldsymbol{x}-\boldsymbol{x}'\|_W$.  

Then the associated dual norm is $\|\boldsymbol{g}\|_{W^{-1}}=(\sum_{j=1}^p g_j^2/w_j)^{1/2}$, with the convention that $g_j^2/w_j=0$ if $w_j=0$ and $g_j=0$. 
For a target function $f(\cdot)$, its Lipschitz constant under $d_W$ is given by:
\begin{align}
    L(W) = \sup_{x\neq x'} \frac{|f(\boldsymbol{x})-f(\boldsymbol{x}')|}{d_W(\boldsymbol{x},\boldsymbol{x}')}.
\end{align}

\noindent\textbf{Theorem 2a (Optimal cure weights).}
For the cure probability $S(H\mid \boldsymbol{x},z)=\sigma(\beta_{z0}+\boldsymbol{\beta}_z^\top \boldsymbol{x})$, 
the Lipschitz constant satisfies 
\begin{align}
    L_{\text{cure}}(W) \le \tfrac14 \|\boldsymbol{\beta}_z\|_{W^{-1}}.
\end{align}
Under the equal-scale constraint $\sum_{j=1}^p w_j=c$, the bound is minimized by the \textit{cure weight matrix}
\begin{align}
     W_{\text{c},z}
    &= \mathrm{diag}\!\left(\frac{c\,|\beta_{z1}|}{\sum_{k=1}^p|\beta_{zk}|},\,
                           \ldots,\,
                           \frac{c\,|\beta_{zp}|}{\sum_{k=1}^p|\beta_{zk}|}\right)\\\notag
    &= \mathrm{diag}|\boldsymbol{\beta}_z|.
\end{align}


\medskip
\noindent\textbf{Theorem 2b (Optimal time weights).}
Let $\eta(\boldsymbol{x})=\lambda_0+\boldsymbol{\lambda}_z^\top \boldsymbol{x}$ and assume 
$(\log T\mid \boldsymbol{X}=\boldsymbol{x})\sim\mathcal N(\eta(\boldsymbol{x}),\sigma^2)$. Denote the non-cured expected event time by $m(\boldsymbol{x})=\mathbb E[T\mid T<H,\boldsymbol{X}=\boldsymbol{x}]$. Then there exists a finite constant $C_{\text{time}}$ 
such that
\begin{align}
    L_{\text{time}}(W) \le C_{\text{time}} \|\boldsymbol{\lambda}_z\|_{W^{-1}}.
\end{align}
Under the equal-scale constraint $\sum_{j=1}^p w_j=c$, the bound is minimized by the \textit{time weight matrix}
\begin{align}
    W_{\text{t},z} = \mathrm{diag}|\boldsymbol{\lambda}_z|.
\end{align}

In practice, the distance metric must be defined uniformly across both treatment and control groups, as matching requires a common rule. Although the averaged weights \ref{eq: distance metric} are not individually optimal within each arm, averaging provides a natural compromise. The resulting bias bound remains controlled (via the maximum across arms), and consistency is preserved.

\section{Simulation Experiments}

\subsection{Setup}

We simulate data using cure and time-to-event models, where treatment affects both cure probability and event timing. Effects are potentially heterogeneous, varying with covariates $\bm{X} \in \mathbb{R}^p$, which include $p_{\text{bin}}\le p$ binary and $p_{\text{cont}} \le p$ continuous features sampled from $\mathsf{Bern}(0.5)$ and $\mathsf{N}(0,1)$, respectively. Treatment assignment follows a logistic model: 
$P(Z = 1 \mid \bm{X}) = \text{expit}(\bm{X}^\top \beta_Z)$. Cure status is modeled separately for each treatment group:
$P(E=1 \mid Z=z, \bm{X}) = \text{expit}(\bm{X}^\top \beta_{E_z})$. Among non-cured individuals ($E=1$), the log event time follows: $(\log T \mid Z=z, \bm{X}) \sim \mathsf{N}(\bm{X}^\top \beta_{T_z}, 1)$, with right truncation at a finite horizon $H$. Independent right-censoring times are drawn uniformly from $[0, 1.5H]$, inducing nontrivial censoring rates. The observed time $Y = \min\{T, C, H\}$, with event indicator $\delta = \mathbbm{1}\{T < \min(C, H)\}$. The simulated dataset of 20,000 samples is split into a training set (7,000) and an estimation set (13,000).

The use of treatment group specific coefficients ($\beta_{E_1}$, $\beta_{E_0}$) and ($\beta_{T_1}$, $\beta_{T_0}$) allows for various treatment effect modifier setups to be explored. We consider four distinct simulation settings where treatment effect modifiers are present in: \textbf{1.} the cure model only; \textbf{2.} the time model only; \textbf{3.} both models on disjoint covariates; \textbf{4.} both models on overlapping covariates. Additional details are provided in supplementary materials.

\subsection{Comparison of Methods}
We compare out method to several matching and model-based estimators. Matching-based estimators include K-nearest neighbor (KNN) matching with Euclidean distance \citep{prijs2007nearest}, propensity score matching \citep{austin2014use, lu2018testing}, prognostic score matching based on Cox models \citep{anand1999evaluation}, outcome-based feature selection matching \citep{kroupa2019assessing}, and optimal full matching on the propensity score scale\citep{austin2015optimal, lee2020empirical}. As a non-matching baseline, we also fit separate Cox PH models within treatment groups \citep{cox1972regression, kalbfleisch2002statistical} to estimate group-specific survival functions for HTE estimation.

As benchmarks, we include \textbf{two oracle estimators}: 1. A full oracle with access to true model coefficients as weights, cure status, and event times under both arms, enabling optimal matching and HTE estimation without censoring. Its error reflects only intrinsic randomness from variability of Bernoulli and event time distribution. 2. A partial oracle that uses the true model parameters for matching but relies on observed, potentially censored outcomes. It reflects the best possible performance given known covariate importance, isolating model fitting error. Comparing the the oracle estimators and the proposed method allows us to decompose error into irreducible outcome variability and estimation error from model fitting and censoring.

For our method, in addition to separate matching using cure-based (for HTE on cure probability) or time-based weights (for HTEs on CMET), we also evaluate a combined version of weights that averages both. HTE estimation accuracy is evaluated using mean absoluate error (MAE) against ground-truth effects across 300 simulations, with standard deviations and distribution of individual-level errors for assessing robustness.

\vspace{-0.8em}
\subsection{Results}
\vspace{-0.4em}

\begin{table*}
    \centering
    \caption{MAE ($\pm$SD) of HTE Estimation on Cure Probability ($\times 100$) and CMET Across Four Settings ($K=50)$ }
    \resizebox{\textwidth}{!}{
    \begin{tabular}{l|cc|cc|cc|cc}
        \hline
        \textbf{Method} & \multicolumn{2}{c|}{\textbf{Cure Only}} & \multicolumn{2}{c|}{\textbf{Time Only}} & \multicolumn{2}{c|}{\textbf{Both+Indep}} & \multicolumn{2}{c}{\textbf{Both+Overlap}} \\
        & Cure & Time & Cure & Time & Cure & Time & Cure & Time \\
        \hline
        Oracle                & 6.6 $\pm$ 0.2 & 12.3 $\pm$ 1.3 & 6.3 $\pm$ 0.3 & 22.0 $\pm$ 1.8 & 6.7 $\pm$ 0.4 & 16.0 $\pm$ 1.9 & 6.8 $\pm$ 0.3 & 18.2 $\pm$ 1.5 \\
        Partial Oracle        & 7.7 $\pm$ 0.3 & 26.8 $\pm$ 1.5 & 7.6 $\pm$ 0.2 & 32.6 $\pm$ 1.5 & 8.0 $\pm$ 0.3 & 31.3 $\pm$ 1.8 & 8.3 $\pm$ 0.3 & 33.4 $\pm$ 1.5 \\
        \textbf{MCM KNN}      & 7.9 $\pm$ 0.3 & 26.8 $\pm$ 1.0 & 7.6 $\pm$ 0.2 & 33.3 $\pm$ 1.5 & 8.3 $\pm$ 0.3 & 31.4 $\pm$ 1.3 & 8.7 $\pm$ 0.3 & 38.5 $\pm$ 1.3 \\
        MCM KNN combined      & 8.0 $\pm$ 0.3 & 27.6 $\pm$ 1.1 & 7.8 $\pm$ 0.2 & 39.4 $\pm$ 1.6 & 8.7 $\pm$ 0.3 & 33.4 $\pm$ 1.2 & 8.9 $\pm$ 0.3 & 45.9 $\pm$ 1.5 \\
        Feature Selection KNN     & 8.2 $\pm$ 0.3 & 28.0 $\pm$ 1.1 & 8.0 $\pm$ 0.2 & 43.2 $\pm$ 1.8 & 9.4 $\pm$ 0.5 & 38.1 $\pm$ 2.5 & 9.4 $\pm$ 0.4 & 41.5 $\pm$ 2.6 \\
        Euclidean KNN         & 9.9 $\pm$ 0.3 & 29.8 $\pm$ 1.1 & 8.3 $\pm$ 0.2 & 48.7 $\pm$ 1.6 & 10.5 $\pm$ 0.4 & 52.1 $\pm$ 1.1 & 10.8 $\pm$ 0.4 & 55.6 $\pm$ 1.3 \\
        Propensity Score KNN  & 17.2 $\pm$ 0.2 & 37.1 $\pm$ 1.3 & 9.6 $\pm$ 0.4 & 91.9 $\pm$ 1.2 & 19.2 $\pm$ 0.2 & 70.3 $\pm$ 0.9 & 19.2 $\pm$ 0.2 & 88.7 $\pm$ 0.9 \\
        Prognostic Score KNN  & 13.6 $\pm$ 0.5 & 38.3 $\pm$ 1.8 & 8.4 $\pm$ 0.4 & 92.9 $\pm$ 1.6 & 14.5 $\pm$ 0.5 & 60.4 $\pm$ 2.1 & 15.4 $\pm$ 0.4 & 68.9 $\pm$ 2.4 \\
        Cox Model (no match)  & 5.2 $\pm$ 0.6 & 36.2 $\pm$ 2.2 & 8.2 $\pm$ 0.7 & 78.7 $\pm$ 3.1 & 7.7 $\pm$ 0.7 & 58.0 $\pm$ 1.3 & 8.6 $\pm$ 0.7 & 67.5 $\pm$ 1.5 \\
        \hline
    \end{tabular}
    }
    \label{tab:combined_mae}
\end{table*}

Table~\ref{tab:combined_mae} report the MAE when estimating HTEs on cure probability and conditional mean event time (CMET) across four simulation settings. Our method consistently yields MAE slightly above that of the full and partial oracle baselines. For instance, whereas the full oracle achieves MAEs around 6.5, our method typically ranges from 7.5 to 8.5. The small gap suggests most residual error, reflected by the oracle MAE, is irreducible due to data variability rather than estimation. The small difference between our approach and the partial oracle, reflecting the estimation error from the model fitting process, further indicates our learned distance metric closely approximates the optimal matching structure. Given the presence of censoring and unknown cure status in time-to-event data, these results indicate our method closely approximates optimal performance in these scenarios.

In comparison with other matching-based methods, our approach achieves substantially lower error across all settings. Standard Euclidean matching fails to distinguish relevant covariates, while propensity score and prognostic score matching do not perform well in estimating both the effects, even they address confounding. The improvement is most evident in settings where treatment effect modifiers influence both the cure and time components. In such cases, feature selection-based matching shows a notable increase in MAE. We also evaluate a combined-weights matching approach that averages the two optimal weights. Since the combined weights do not directly target either estimand, this approach is comparable in design to general-purpose matching strategies. However, it consistently performs slightly better than these baselines, highlighting the benefit of the learned weights even when not explicitly aligned with the estimands. We also include Cox proportional hazards models as a non-matching baseline. When treatment effect modifiers influence only the event probability, Cox performs well, but when modifiers affect event timing, Cox fails to capture heterogeneity and yields substantially higher errors, particularly in estimating time effects. These results suggest that Cox models perform well when their assumptions align with the underlying data, but their limitations become apparent when heterogeneity spans both event probability and timing.

\vspace{-0.8em}
\begin{figure}[htbp]
    \centering
    \includegraphics[width=0.7\columnwidth]{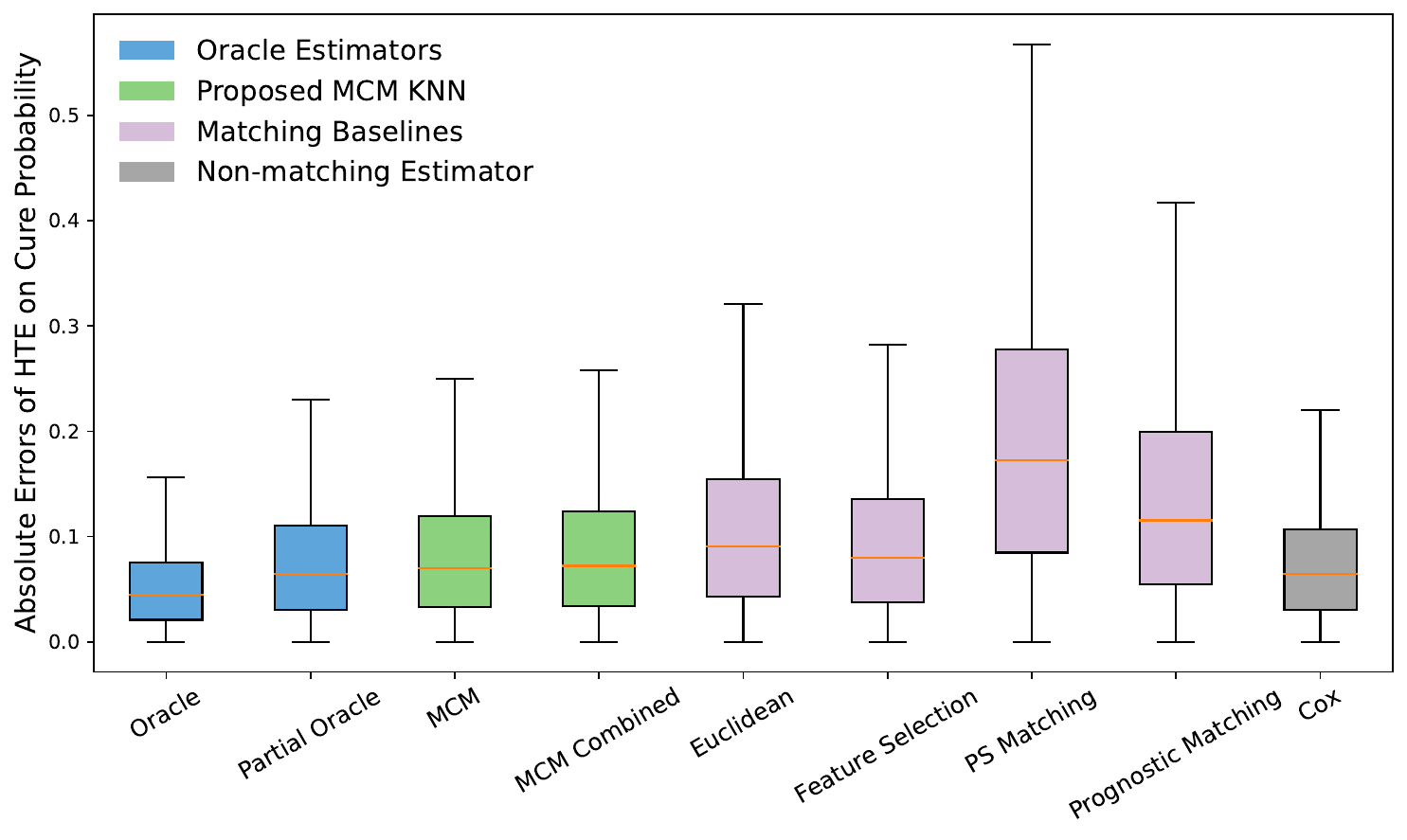}  
    \caption{Absolute HTE Estimation Error on Cure Probability.}
    \label{fig:hte_error_boxplot}
\end{figure}

\vspace{-0.8em}

Figure~\ref{fig:hte_error_boxplot} shows the distribution of absolute errors in estimating individual-level HTEs on cure probability, computed on the estimation set in a single simulation run. This setting involves effect modifiers in both the cure and time components, with disjoint covariates. As expected, oracle estimators achieve the lowest error and variability. Our proposed method (MCM-based matching) ranks closely behind the two oracle baselines.
Compared to Euclidean matching, which is unbiased but agnostic to outcome, our method yields lower mean error and variance. Performance is comparable to feature selection-based matching and the Cox model without matching, though the latter two methods are observed to vary more significantly across different simulation settings as shown in table~\ref{tab:combined_mae}.

Choosing $K$ in KNN matching trades off bias and variance. Smaller $K$ gives precise but variable estimates; larger $K$ improves stability but risks poor matches. We use $K=50$ to balance these factors in experiments. In practice, $K$ should be guided by cohort size, feature variability, and domain knowledge, since standard validation is not feasible. Additional experimental results are shown in the supplementary materials.

\section{Real-World Data Experiments}
To illustrate our method's real-world value, where assumptions may not hold and true models are unknown, we analyze a multi-center cohort of patients with acute lymphoblastic leukemia who received allogeneic stem cell transplantation\citep{ma2021integrative}. The treatment of interest compares haploidentical transplantation (Haplo-SCT) versus matched sibling donor transplantation (MSDT). A total of 1157 patients are included, with follow-up time up to 4193 days.

\vspace{-0.8em}
\subsection{Study Design and Setup}
\vspace{-0.3em}
Among four recorded outcomes, we focus on \textbf{leukemia-free survival (LFS)}, which reflects both relapse and death and serves as a proxy for long-term cure. We define cure at $H = 10$ years (3650 days). The treatment groups are imbalanced, with 911 patients in the treated group and 246 in the control group (approximately 78\% vs 22\%). Censoring rates are 18.1\% and 25.2\% respectively. Baseline covariates include demographics, donor-recipient characteristics, minimal residual disease status (MRD), HLA mismatch, and complete remission status, totaling 8-9 key predictors.We used 35\% of the data for training the mixture cure model and 65\% for HTE estimation, with $K = 35$ neighbors used in matching.

\vspace{-0.8em}
\subsection{HTE Estimation Results}
\vspace{-0.3em}

%

The estimated HTE distributions vary substantially across matching methods, as shown in Figure~\ref{fig:leukemia_cure_violin}.
The proposed method (MCM) produces a concentrated and symmetric distribution of estimated HTEs, with a mean around 0.05 and an interquartile range (IQR) that aligns with the empirical difference in cure probabilities between two arms.  Euclidean distance-based matching yields a symmetric but highly dispersed distribution, while propensity score (PS) matching produces tightly concentrated estimates with limited variability. Prognostic-score matching results in a bimodal and partially negative distribution, potentially due to misalignment between the prognostic model and the cure-related outcome. Feature selection-based matching produces skewed and unstable estimates. 



Analysis of pairwise correlation of estimated HTEs on cure probability shows that the separate-weight MCM and combined-weight MCM variants exhibit strong positive correlation (0.799), 
whereas correlation between separate-weight MCM and Euclidean matching is positive but moderate (0.346), suggesting the methods yield directionally aligned but distinct estimates.

Overall our results suggest that Haplo-SCT improves LFS rates by approximately 5\% on average compared to MSDT, which is clinically meaningful but more conservative than estimates from comparator methods. The difference may arise because our method is designed to estimate distinct effects on cure probability and survival time, unlike the comparators. As a result, comparator methods may conflate treatment effects on LFS probability with effects on short-term survival times, which are also important but clinically distinct. Although ground-truth HTEs are not available in this real-world setting, our theoretical and simulation results strongly suggest that our approach is likely to be more accurate in this dataset due to the presence of a large cured population.

\begin{figure}[htbp]
\centering
\includegraphics[width=0.7\linewidth]{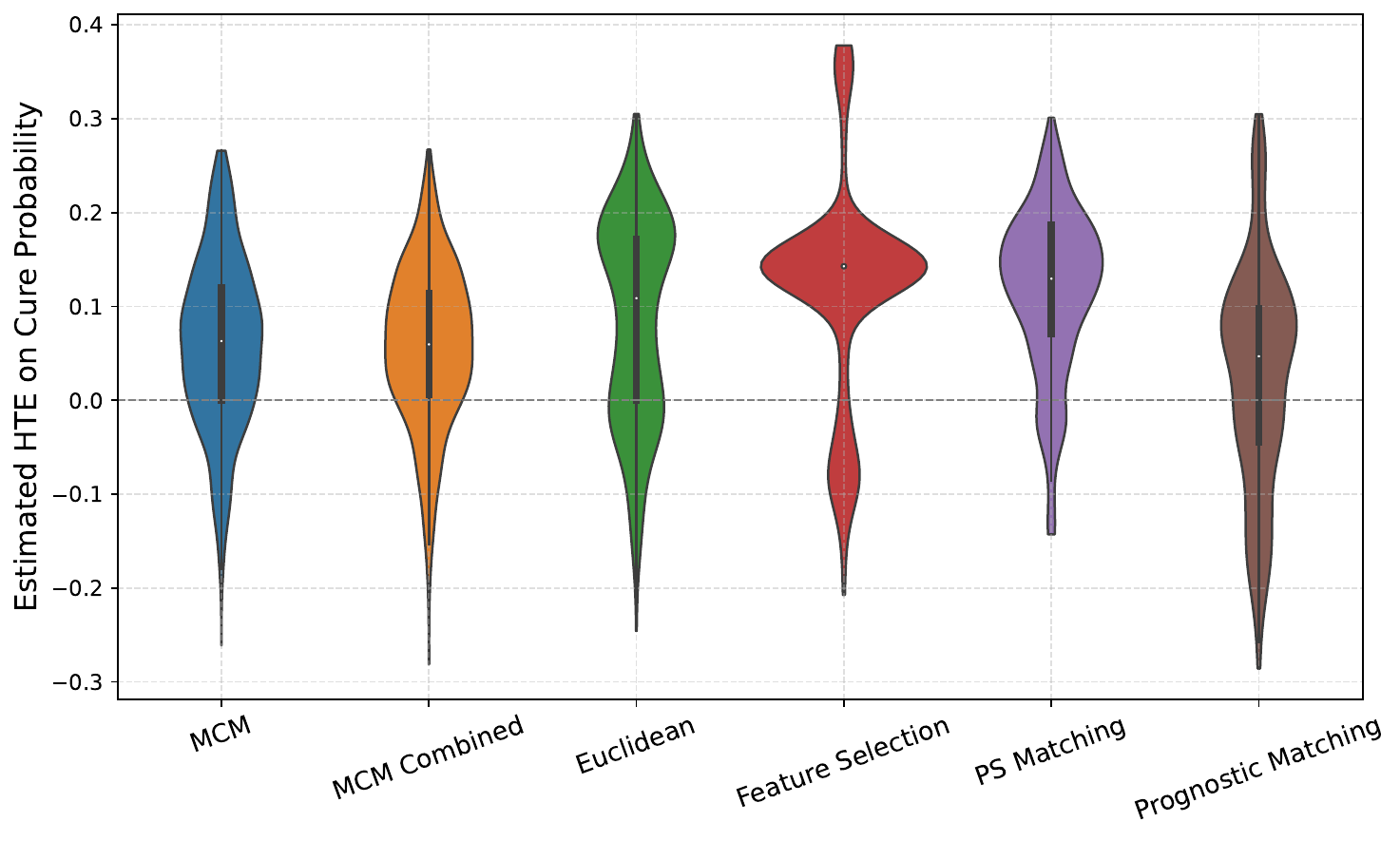}
\caption{
Distribution of Estimated HTEs on Cure Probability under Leukemia-Free Survival (LFS). 
}
\label{fig:leukemia_cure_violin}
\end{figure}
\vspace{-1em}

\section{Conclusion}
We propose a novel matching framework for separately estimating HTE on cure probability and event timing, leveraging covariate-specific distances derived from a mixture cure model. Our method is theoretically grounded, achieves consistent estimation via matched Kaplan-Meier curves, and performs well across simulations. However, it requires further validation on diverse real-world datasets and may remain sensitive to model misspecification despite incorporating parametric and non-parametric elements.

\bibliographystyle{apalike}
\bibliography{ref}

\newpage
\appendix
\onecolumn

\begin{center}
    {\LARGE \textbf{Supplementary Materials}}
\end{center}
\vspace{1em}

\section{MISSING PROOFS}

\subsection{Proof of Proposition 1}
In this section, we present the assumptions and detailed proof of \textbf{Proposition 1}.

\begin{proposition}\label{prop:identification}
    If the consistency, unconfoundedness, positivity and non-informative censoring assumptions hold, then the two estimands can be expressed as
    \begin{equation}\label{eq:cure-estimand}
        \pi(\bm{x})  = S_1(H \mid \bm{x}) - S_0(H \mid \bm{x}),
    \end{equation}
    \begin{equation}\label{eq:cmst-estimand}
        \begin{aligned}
            \Delta(\bm{x}) &= \frac{\int_0^H S_1(t|\bm{x}) \, dt -H S_1(H|\bm{x})}{1 - S_1(H|\bm{x})} \\
                      &- \frac{\int_0^H S_0(t|\bm{x}) \, dt -H S_0(H|\bm{x})}{1 - S_0(H|\bm{x})}.
        \end{aligned}        
    \end{equation}
\end{proposition}

\textbf{Proof.}
Suppose we have the following assumptions:
\begin{assumption}[Consistency]\label{ass:consistency}
    For all $i$, 
    \begin{align}
        T_i &= Z_i T_i(1) + (1 - Z_i) T_i(0) \quad \\
        C_i &= Z_i C_i(1) + (1 - Z_i) C_i(0).
    \end{align}
\end{assumption}

\begin{assumption}[Unconfoundedness]\label{ass:unconfoundedness}
    \begin{equation}
        \{T_i(1), T_i(0), C_i(1), C_i(0)\} \indep Z_i \mid \bm{X}_i.
    \end{equation}
\end{assumption}

\begin{assumption}[Positivity]\label{ass:positivity}
    For all $x$, 
    \begin{equation}
        0 < \Pr(Z_i = 1 \mid \bm{X}_i = \bm{x}) < 1
    \end{equation}
\end{assumption}

\begin{assumption}[Non-informative censoring]\label{ass:non-informative-censoring}
    \begin{equation}
        \{T_i(1), T_i(0)\} \indep \{C_i(1), C_i(0)\} \mid \bm{X}_i
    \end{equation}
\end{assumption}

The proof of \textbf{Proposition 1} can be shown as below:
\begin{equation}
\begin{aligned}
            \pi(\bm{x}) 
        &= \Pr(E_i(1)=0 \mid \bm{X} = \bm{x}) - \Pr(E_i(0)=0 \mid \bm{X} = \bm{x}) \\
        &= \Pr(T(1) > H \mid \bm{X} = \bm{x}) - \Pr(T(0) > H \mid \bm{X} = \bm{x}) \\
        &= \Pr(T(1) > H \mid Z = 1, \bm{X} = \bm{x}) - \Pr(T(0) > H \mid Z = 0, \bm{X} = \bm{x}) \\
        &= \Pr(T > H \mid Z = 1, \bm{X} = \bm{x}) - \Pr(T > H \mid Z = 0, \bm{X} = \bm{x}) \\
        &= S_1(H \mid \bm{x}) - S_0(H \mid \bm{x}), \\
\end{aligned}
\end{equation}

Since we can derive the equation below: 
\begin{equation}\label{eq:time expect}
    \begin{aligned}
        E[T_i(z) \mid T_i(z) < H, \bm{X} = \bm{x}] 
        &= \int t \cdot p(T = t \mid T < H, Z = z, \bm{X} = \bm{x}) \, dt \\
        &= \int t \cdot \frac{p(T = t, T < H\mid Z = z, \bm{X} = \bm{x})}{p(T < H\mid Z = z, \bm{X} = \bm{x})} \, dt \\
        &= \int t \cdot \frac{p(T = t\mid Z = z, \bm{X} = \bm{x}) \cdot p(T < H \mid T = t, Z = z, \bm{X} = \bm{x})}{p(T < H\mid Z = z, \bm{X} = \bm{x})} \, dt \\
        &= \frac{\int_0^H t \cdot p(T = t\mid Z = z, \bm{X} = \bm{x}) \, dt}{\int_0^H p(T = t\mid Z = z, \bm{X} = \bm{x}) \, dt} \\
        &= \frac{\int_0^H t \, dS_z(t \mid \bm{x})}{\int_0^H dS_z(t \mid \bm{x})} \\
        &= \frac{-H S_z(H \mid \bm{x}) + \int_0^H S_z(t \mid \bm{x}) \, dt}{1 - S_z(H \mid \bm{x})}
    \end{aligned}
\end{equation}

then the second estimand can be expressed by:
\begin{equation}
\begin{aligned}
    \Delta(\bm{x}) 
        &= \mathbb{E}[T_i(1) \mid T_i(1) < H, \bm{X} = \bm{x}] - \mathbb{E}[T_i(0) \mid T_i(0) < H, \bm{X} = \bm{x}] \\
        &= \frac{\int_0^H S_1(t|\bm{x}) \, dt -H S_1(H|\bm{x})}{1 - S_1(H|\bm{x})} 
                      - \frac{\int_0^H S_0(t|\bm{x}) \, dt -H S_0(H|\bm{x})}{1 - S_0(H|\bm{x})}
\end{aligned}
\end{equation}

\subsection{Mixture Cure Model}
With a pre-specified time horizon $H$, the survival function modeled by a mixture cure approach can be written as: 
\begin{equation}
    \begin{aligned}
        S(t) &= \Pr(cure) + [1-\Pr(cure)]S_N(t) \\
        &= \Pr(E=0) + [1 - \Pr(E=0)] \Pr( T>t | T \leq H) \\
        &= \Pr(T>H) + \Pr(T \leq H) \Pr( T>t | T \leq H) \\
        &= \Pr(T>t)
    \end{aligned}
\end{equation}
where $\Pr(cure)$ is the probability of being cured, and $S_N(t)$ is the survival function given the patient is non-cured.

\textbf{Likelihood.}
Let $T$ denote the true event time, $C$ the censoring time, and $Y = \min(T, C, H)$ the observed time under a fixed horizon $H$ with event indicator $\delta = \mathbb{I}(T < C, T < H)$.
The covariates are denoted by $\boldsymbol{x}$, and the latent cure indicator by 
$E = \mathbb{I}(T < H)$, where $E = 1$ indicates a non-cured individual and $E = 0$ indicates a cured individual.

We denote:
\begin{itemize}
    \item $f_T(t \mid \boldsymbol{x})$ : the probability density function (PDF) of event time $T$ conditional on covariates $\boldsymbol{x}$;
    \item $F_T(t \mid \boldsymbol{x}) = \Pr(T \le t \mid \boldsymbol{x})$ : the cumulative distribution function (CDF) of $T$;
    \item $S_T(t \mid \boldsymbol{x}) = 1 - F_T(t \mid \boldsymbol{x}) = \Pr(T > t \mid \boldsymbol{x})$ : the survival function of $T$;
    \item $f_C(c \mid \boldsymbol{x})$ : the density function of the censoring time $C$ given $\boldsymbol{x}$;
    \item $S_C(c \mid \boldsymbol{x}) = \Pr(C > c \mid \boldsymbol{x})$ : the survival function of $C$.
\end{itemize}

The density(likelihood) of data $p(y,\delta|\bm{x})$ can be summarized as follows:
\begin{equation}\label{eq: density}
    \begin{aligned}
        p(y, &\delta | \bm{x}) 
        = p(y, \delta=1|\bm{x}) + p(y, \delta=0|\bm{x}) \\
        &= p(E=1|\bm{x}) \mathbb{I}(y<H) \frac{f_T(y|\bm{x})}{F_T(H|\bm{x})} S_C(y|\bm{x}) 
        + p(E=1|\bm{x}) \mathbb{I}(y<H) \left( 1 - \frac{F_T(y|\bm{x})}{F_T(H|\bm{x})} \right) f_C(y|\bm{x}) \\
        &+ p(E=0|\bm{x}) \mathbb{I}(y<H) f_C(y|\bm{x}) 
        + p(E=0|\bm{x}) \mathbb{I}(y=H) S_C(y|\bm{x}) 
    \end{aligned}    
\end{equation}

\textbf{Proof.}
For the first term in ~\ref{eq: density},
\begin{equation}
\begin{aligned}
        p(y, \delta=1 | \bm{x}) 
    &= p(y, \delta=1, E=1 | \bm{x}) + p(y, \delta=1, E=0 | \bm{x}) \\
    &= p(E=1|\bm{x}) p(y, \delta=1 | \bm{x}, E=1) + 0 \\
    &= p(E=1|\bm{x}) p(T=y, y<H, y<C | \bm{x}, T<H) \\
    &= p(E=1|\bm{x}) \mathbb{I}(y<H) p(T=y, y<C | \bm{x}, T<H) \\
    &= p(E=1|\bm{x}) \mathbb{I}(y<H) p(T=y| T<H, \bm{x}) p(C>y | \bm{x}, T<H) \\
    &= p(E=1|\bm{x}) \mathbb{I}(y<H) \frac{p(T=y, T<H |\bm{x})}{p(T<H|\bm{x})} p(C>y | \bm{x}) \\
    &= p(E=1|\bm{x}) \mathbb{I}(y<H) \frac{p(T=y |\bm{x})}{p(T<H|\bm{x})} p(C>y | \bm{x}) \\
    &= p(E=1|\bm{x}) \mathbb{I}(y<H) \frac{f_T(y |\bm{x})}{F_T(H|\bm{x})} S_C(y|\bm{x})
\end{aligned}
\end{equation}

and for the second term in ~\ref{eq: density},
\begin{equation}\label{eq:density_second}
\begin{aligned}
    p(y, \delta=0, | \bm{x}) 
    &= p(y, \delta=0, E=1 | \bm{x}) + p(y, \delta=0, E=0 | \bm{x}) \\
    &= p(E=1 | \bm{x})p(y, \delta=0 | \bm{x}, E=1) + p(E=0 | \bm{x})p(y, \delta=0 | \bm{x}, E=0) \\
    &= p(E=1 | \bm{x})p(y, \delta=0 | \bm{x}, T<H) + p(E=0 | \bm{x})p(y, \delta=0 | \bm{x}, T\geq H)
\end{aligned}
\end{equation}

For the $p(y, \delta=0 | \bm{x}, T<H)$ in the first term of ~\ref{eq:density_second}:
\begin{align}
    p(y, \delta=0 | \bm{x}, T<H) 
    &= p(C=y, T>y, H>y | \bm{x}, T<H) + p(H=y, T\geq y, C\geq y | \bm{x}, T<H) \\
    &= \mathbb{I}(y<H)\frac{p(T>y, T<H |\bm{x})}{p(T<H|\bm{x})}p(C=y|\bm{x}) + \mathbb{I}(y=H)p(T\geq H, C \geq H | \bm{x}, T<H) \\
    &= \mathbb{I}(y<H)\left( 1 - \frac{F_T(y|\bm{x})}{F_T(H|\bm{x})}\right) f_C(y|\bm{x}) + 0
\end{align}

For the $p(y, \delta=0 | \bm{x}, T\geq H)$ in the second term of ~\ref{eq:density_second}:
\begin{equation}
\begin{aligned}
    p(y, \delta=0 | \bm{x}, T\geq H) 
    &= p(T > y, C = y, H > y \mid \boldsymbol{x}, E = 0) + p(T \geq y, C \geq y, H = y \mid \boldsymbol{x}, E = 0) \\
    &= \mathbb{I}(y < H) \,
    \frac{p(T > y, T \geq H \mid \boldsymbol{x})}{p(T \geq H \mid \boldsymbol{x})} \,
    p(C = y \mid \boldsymbol{x}) 
    + \mathbb{I}(y = H) \frac{p(T \geq y, T \geq H \mid \boldsymbol{x})}{p(T \ge H \mid \boldsymbol{x})}p(C \geq y \mid \boldsymbol{x})  \\
    &= \mathbb{I}(y < H) f_C(y \mid \boldsymbol{x}) + \mathbb{I}(y = H) S_C(y \mid \boldsymbol{x})
\end{aligned}
\end{equation}

By plugging in the terms derived above, we can obtain the likelihood function as states in ~\ref{eq: density}.

\subsection{Proof of Theorem 1}
\textbf{Setup.}
Let $S(t \mid \boldsymbol{x}) := \Pr(T > t \mid \boldsymbol{X} = \boldsymbol{x})$ be the conditional survival curve given covariate $\boldsymbol{x}$. Let $\boldsymbol{X}_1, \dots, \boldsymbol{X}_n$ be units drawn i.i.d. from a distribution $P$. For each $i \in \{1, 2, \dots, n\}$, let $T_i$ be the random event time with survival function $S(t \mid \boldsymbol{X}_i)$ and let $C_i$ be the random censor time independent of $T_i$. Define the observed time $Y_i = \min\{T_i, C_i\}$ and the event indicator $\delta_i = \mathbb{I}(T_i < C_i)$.

Consider a matching group of size $m$, $\{\boldsymbol{X}_{k_1}, \dots, \boldsymbol{X}_{k_m}\} \subset \{\boldsymbol{X}_1, \dots, \boldsymbol{X}_n\}$, selected to estimate $S(t \mid \boldsymbol{x})$. Conventionally, these $m$ units are selected such that they have the closest distance to $\boldsymbol{x}$ among all units under a certain distance metric $d(\cdot, \cdot)$. Let $\widehat{S}_{mn}(t)$ be the Kaplan-Meier estimate with $\{(Y_{k_i}, \delta_{k_i})\}_{i=1}^m$.

Define two functionals
$$F(S) := S(H), \quad G(S) := \int_0^H S(u)\,\mathrm{d}u.$$ The following theorem establishes that, under certain conditions, $F(\widehat{S}_{mn})$ is consistent to $F(S(\cdot \mid \boldsymbol{x}))$ and $G(\widehat{S}_{mn})$ is consistent to $G(S(\cdot \mid \boldsymbol{x}))$.

\textbf{Theorem 1.}
Let $\boldsymbol{x}$ be in the support of $\bm{X}$, \emph{i.e.}, for any $\epsilon > 0$, $P(d(\boldsymbol{X}, \boldsymbol{x}) < \epsilon) > 0$. Assume $S(t\mid \boldsymbol{x})$ is locally Lipschitz in $\boldsymbol{x}$ uniformly in $t$, \emph{i.e.}, there exists a neighborhood of $\boldsymbol{x}$, denoted by $\mathcal{N}_{\boldsymbol{x}}$, and a constant $\tau > H$ such that for all $\boldsymbol{x}', \boldsymbol{x}''$ in $\mathcal{N}_{\boldsymbol{x}}$ and $t \in [0, \tau]$, $|S(t \mid \boldsymbol{x}') - S(t \mid \boldsymbol{x}'')| \leq L \cdot d(\boldsymbol{x}', \boldsymbol{x}'')$. Then, 
    $$F(\widehat{S}_{mn}) \stackrel{p}{\rightarrow} F(S(\cdot \mid \boldsymbol{x})), \quad G(\widehat{S}_{mn}) \stackrel{p}{\rightarrow} G(S(\cdot \mid \boldsymbol{x})),$$
    as $m, n \rightarrow \infty$ with $m = o(n)$.

\textbf{Corollary 1.}
The two estimators defined in Section 4.3 as shown below, which are functions of $F$ and $G$, are also consistent.

\begin{align}\label{eq:hte_cure_estimator}
    \pi(\mathbf{x}_i) = \hat S_{M_1}(H) - \hat S_{M_0}(H),
\end{align}
\begin{align}
    \Delta(\bm{x}_i) = \frac{\int_0^H \hat S_{M_1}(t) \, dt -H \hat S_{M_1}(H)}{1 - \hat S_{M_1}(H)} - \frac{\int_0^H \hat S_{M_0}(t) \, dt -H \hat S_{M_0}(H)}{1 - \hat S_{M_0}(H)}
\end{align}

\textbf{Proof.}
We will prove \textbf{Theorem 1} via three lemmas. The first lemma shows that the matching group gets closer to $\boldsymbol{x}$ as sample size increases.

\textbf{Lemma 1.}
    Let $d_{mn} = \max_{i} d(\boldsymbol{X}_{k_i}, \boldsymbol{x})$, where $\boldsymbol{x}$ is in the support of $\bm{X}$. Then, $d_{mn} \stackrel{p}{\rightarrow} 0.$

\textbf{Proof of Lemma 1.}
    Define $d_i = d(\boldsymbol{X}_i, \boldsymbol{x})$. Then, $d_{mn}$ is the $m$-th order statistic of $d_1, \dots, d_n$, denoted by $d_{(m)}$. Let $Q(c) = P(d(\boldsymbol{X}, \boldsymbol{x}) \leq c)$ be the CDF of the distribution of $d(\boldsymbol{X}, \boldsymbol{x})$. For any $\alpha \in (0, 1)$, the sample quantile $d_{(\lfloor \alpha n \rfloor)}$ converges in probability to the $\alpha$-quantile $Q^{-1}(\alpha)$, because the sample quantile $d_{(\lfloor \alpha n \rfloor)}$ is an $M$-estimator with estimating equation $\psi_i(c) = \mathbb{I}(d_i \leq c) - \alpha$.
    Therefore, $d_{(\lfloor \alpha n \rfloor)} = Q^{-1}(\alpha) + o_p(1)$ as $n \rightarrow \infty$. Since $m = o(n)$, for any $\alpha \in (0, 1)$, $m < \alpha n$ when $n$ is sufficiently large. Therefore, $d_{(m)} \leq d_{(\lfloor \alpha n \rfloor)} = Q^{-1}(\alpha) + o_p(1)$ as $n \rightarrow \infty$. Since $\alpha$ is arbitrarily chosen, and $\lim_{\alpha \rightarrow 0} Q^{-1}(\alpha) = 0$ because $\boldsymbol{x}$ is in the support of $\bm{X}$, we have $d_{(m)} = o_p(1)$ which yields $d_{mn} \stackrel{p}{\rightarrow} 0$.

\vspace{1em}
The second lemma shows that the Kaplan--Meier estimator is strongly consistent with the mixture survival function, 
as established in \citet{Stute1993}.

\textbf{Lemma 2.}
    Define $\bar{S}(t) = \frac{1}{m}\sum_{i=1}^m S(t \mid \boldsymbol{X}_{k_i})$. Then, 
    $$\sup_{0 \leq t \leq \tau} |\widehat{S}_{mn}(t) - \bar{S}(t)| \stackrel{p}{\rightarrow} 0$$ as $m, n \rightarrow \infty$ with $m = o(n)$.

\vspace{1em}

The third lemma shows that the mixture survival function $\bar{S}(t)$ does not deviate far from $S(t \mid \boldsymbol{x})$ as long as the covariates $\boldsymbol{X}_{k_i}$ are close to $\boldsymbol{x}$.

\textbf{Lemma 3.}
    If $\boldsymbol{x}$ is in the support of $\bm{X}$ and $S(t\mid \boldsymbol{x})$ is locally Lipschitz uniformly in $t$, then $\bar{S}(t)$ is strongly consistent to $S(t \mid \boldsymbol{x})$ on $[0, \tau]$.

\textbf{Proof of Lemma 3.}
    $$|\bar{S}(t) - S(t \mid \boldsymbol{x})| \leq \frac{1}{m}\sum_{i=1}^m\left|S(t\mid \boldsymbol{X}_{k_i}) - S(t \mid \boldsymbol{x})\right|.$$
    Furthermore, when $d(\boldsymbol{X}_{k_i}, \boldsymbol{x}) < \epsilon$, $$\left|S(t\mid \boldsymbol{X}_{k_i}) - S(t \mid \boldsymbol{x})\right| \leq L d(\boldsymbol{X}_{k_i}, \boldsymbol{x}).$$

    Therefore, for any $\delta  > 0$,
    \begin{align}
        P\left(|\bar{S}(t) - S(t \mid \boldsymbol{x})| \geq \delta\right) &\leq P\left(\frac{1}{m}\sum_{i=1}^m\left|S(t\mid \boldsymbol{X}_{k_i}) - S(t \mid \boldsymbol{x})\right| \geq \delta\right) \\
        &\leq P\left(\frac{1}{m}\sum_{i=1}^m\left|S(t\mid \boldsymbol{X}_{k_i}) - S(t \mid \boldsymbol{x})\right| \geq \delta, d_{mn} < \epsilon\right) + P(d_{mn} \geq \epsilon) \\
        &\leq P\left(\frac{1}{m}\sum_{i=1}^m\left|S(t\mid \boldsymbol{X}_{k_i}) - S(t \mid \boldsymbol{x})\right| \geq \delta \mid d_{mn} < \epsilon\right) + P(d_{mn} \geq \epsilon) \\
        &\leq P(Ld_{mn} \geq \delta \mid d_{mn} < \epsilon) + P(d_{mn} \geq \epsilon) \\
        &\leq P(d_{mn} \geq \delta / L) + P(d_{mn}\geq \epsilon).
    \end{align}

    Since $d_{mn}$ converges in probability to 0, the above probability converges to 0. Therefore, $\bar{S}(t)$ is strongly consistent to $S(t \mid \boldsymbol{x})$ on $[0, \tau]$.

\vspace{1em}
\textbf{Proof of the main theorem}
    With the three lemmas, $\widehat{S}_{mn}(t)$ is strongly consistent to $S(t \mid \boldsymbol{x})$ under the conditions of Theorem 1. By the triangle inequality, 
$$
\left| \widehat{S}_{mn}(t) - S(t \mid \boldsymbol{x}) \right|
\leq \left| \widehat{S}_{mn}(t) - \bar{S}(t) \right| + \left| \bar{S}(t) - S(t \mid \boldsymbol{x}) \right|.
$$

Therefore, for any \( \epsilon > 0 \),
$$
\mathbb{P}\left( \left| \widehat{S}_{mn}(t) - S(t \mid \boldsymbol{x}) \right| > \epsilon \right)
\leq \mathbb{P}\left( \left| \widehat{S}_{mn}(t) - \bar{S}(t) \right| > \frac{\epsilon}{2} \right)
+ \mathbb{P}\left( \left| \bar{S}(t) - S(t \mid \boldsymbol{x}) \right| > \frac{\epsilon}{2} \right).
$$

Since both terms on the right converge to zero as \( n \to \infty \), we conclude:
\[
\mathbb{P}\left( \left| \widehat{S}_{mn}(t) - S(t \mid \boldsymbol{x}) \right| > \epsilon \right) \to 0,
\]
which implies
\[
\widehat{S}_{mn}(t) \xrightarrow{P} S(t \mid \boldsymbol{x}).
\]

\vspace{1em}
\textbf{As a special case,} $\widehat{S}_{mn}(H)$ is consistent to $S(H \mid \boldsymbol{x})$. This proves that $F(\widehat{S}_{mn}) \stackrel{p}{\rightarrow} F(S(\cdot \mid \boldsymbol{x}))$.

Furthermore, 
$$|G(\widehat{S}_{mn}) - G(S(\cdot \mid \boldsymbol{x}))| = \left|\int_0^H \widehat{S}_{mn}(u) - S(u \mid \boldsymbol{x})\,\mathrm{d} u\right| \leq H \sup_{t \in [0, H]} |\widehat{S}_{mn}(t) - S(t \mid \boldsymbol{x})| \stackrel{p}{\rightarrow} 0.$$
This completes the proof that $G(\widehat{S}_{mn}) \stackrel{p}{\rightarrow} G(S(\cdot \mid \boldsymbol{x}))$.

\subsection{Proof of Theorem 2a and 2b}

\paragraph{Setup.} 
For a diagonal weight matrix $W=\mathrm{diag}(w_1,\ldots,w_p)$ with $w_j\ge 0$ and $\sum_{j=1}^p w_j=c>0$, define the weighted norm $\|\boldsymbol{v}\|_W=(\sum_{j=1}^p w_j v_j^2)^{1/2}$ and distance $d_W(\boldsymbol{x},\boldsymbol{x}')=\|\boldsymbol{x}-\boldsymbol{x}'\|_W$.  

Then the associated dual norm is $\|\boldsymbol{g}\|_{W^{-1}}=(\sum_{j=1}^p g_j^2/w_j)^{1/2}$, with the convention that $g_j^2/w_j=0$ if $w_j=0$ and $g_j=0$. 
For a target function $f(\cdot)$, its Lipschitz constant under $d_W$ is given by:
\begin{align}
    L(W) = \sup_{x\neq x'} \frac{|f(\boldsymbol{x})-f(\boldsymbol{x}')|}{d_W(\boldsymbol{x},\boldsymbol{x}')}.
\end{align}

\noindent\textbf{Theorem 2a.}
For the cure probability $S(H\mid \boldsymbol{x},z)=\sigma(\beta_{z0}+\boldsymbol{\beta}_z^\top \boldsymbol{x})$, 
the Lipschitz constant satisfies 
\begin{align}
    L_{\text{cure}}(W) \le \tfrac14 \|\boldsymbol{\beta}_z\|_{W^{-1}}.
\end{align}
Under the equal-scale constraint $\sum_{j=1}^p w_j=c$, the bound is minimized by the \textit{cure weight matrix}
\begin{align}
     W_{\text{c},z}
    &= \mathrm{diag}\!\left(\frac{c\,|\beta_{z1}|}{\sum_{k=1}^p|\beta_{zk}|},\,
                           \ldots,\,
                           \frac{c\,|\beta_{zp}|}{\sum_{k=1}^p|\beta_{zk}|}\right)\\\notag
    &= \mathrm{diag}|\boldsymbol{\beta}_z|.
\end{align}


\textbf{Proof of Theorem 2a.}
The gradient of the cure function is 
\begin{align}
    \nabla_{\boldsymbol{x}} S(H\mid \boldsymbol{x}, z)
    = \sigma'(\beta_{z0}+\boldsymbol{\beta}_z^\top \boldsymbol{x})\,\boldsymbol{\beta}_z,
    \qquad 
    \sigma'(u)=\sigma(u)\{1-\sigma(u)\}\le \tfrac14.
\end{align}

Hence, 
\begin{align}
    \|\nabla_{\boldsymbol{x}} S(H\mid \boldsymbol{x}, z)\|_{W^{-1}}
    \le \tfrac14\,\|\boldsymbol{\beta}_z\|_{W^{-1}}
    = \tfrac14\Big(\sum_{j=1}^p \tfrac{\beta_{zj}^2}{w_j}\Big)^{1/2}.
\end{align}

By the dual-norm inequality (mean value theorem combined with Hölder),
\begin{align}
    |S(H\mid \boldsymbol{x}, z)-S(H\mid \boldsymbol{x}', z)|
    \le \sup_{\boldsymbol{x}} \|\nabla_{\boldsymbol{x}} S(H\mid \boldsymbol{x}, z)\|_{W^{-1}}\,
       d_W(\boldsymbol{x},\boldsymbol{x}'),
\end{align}

which yields the Lipschitz bound
\begin{align}
    L_{\text{cure}}(W)
    \le \tfrac14\Big(\sum_{j=1}^p \tfrac{\beta_{zj}^2}{w_j}\Big)^{1/2}.
\end{align}

To minimize $\sum_j \beta_{zj}^2/w_j$ subject to $w_j\ge0$ and $\sum_j w_j=c$,
apply Cauchy--Schwarz:
\begin{align}
    \Big(\sum_{j=1}^p |\beta_{zj}|\Big)^2
    =\Big(\sum_j \sqrt{w_j}\cdot \tfrac{|\beta_{zj}|}{\sqrt{w_j}}\Big)^2
    \le \Big(\sum_j w_j\Big)\Big(\sum_j \tfrac{\beta_{zj}^2}{w_j}\Big)
    = c\sum_j \tfrac{\beta_{zj}^2}{w_j}.
\end{align}

Thus $\sum_j \beta_{zj}^2/w_j \ge (\sum_j|\beta_{zj}|)^2/c$, 
with equality if and only if $w_j\propto |\beta_{zj}|$. 
Enforcing $\sum_j w_j=c$ fixes the constant of proportionality,
giving the unique minimizer
\begin{align}
    w_j^\star=\frac{c\,|\beta_{zj}|}{\sum_{k=1}^p|\beta_{zk}|}, 
    \qquad j=1,\dots,p,
\end{align}
which yields the optimal cure weight matrix $W_{\text{c},z}$ 
and completes the proof.

\medskip
\noindent\textbf{Theorem 2b (Optimal time weights).}
Let $\eta(\boldsymbol{x})=\lambda_0+\boldsymbol{\lambda}_z^\top \boldsymbol{x}$ and assume 
$(\log T\mid \boldsymbol{X}=\boldsymbol{x})\sim\mathcal N(\eta(\boldsymbol{x}),\sigma^2)$. Denote the non-cured expected event time by $m(\boldsymbol{x})=\mathbb E[T\mid T<H,\boldsymbol{X}=\boldsymbol{x}]$. Then there exists a finite constant $C_{\text{time}}$ 
such that
\begin{align}
    L_{\text{time}}(W) \le C_{\text{time}} \|\boldsymbol{\lambda}_z\|_{W^{-1}}.
\end{align}
Under the equal-scale constraint $\sum_{j=1}^p w_j=c$, the bound is minimized by the \textit{time weight matrix}
\begin{align}
    W_{\text{t},z} = \mathrm{diag}|\boldsymbol{\lambda}_z|.
\end{align}

\textbf{Proof of Theorem 2b.}

Conditional on $T<H$, the non-cured event time follows a truncated log-normal distribution:
\begin{align}
    T \sim \mathrm{LogNormal}(\eta(\boldsymbol{x}), \sigma^2)
    \quad \text{truncated to } (0,H),
\end{align}
where $\eta(\boldsymbol{x}) = \lambda_{z0} + \boldsymbol{\lambda}_z^\top \boldsymbol{x}$.

Its conditional CDF is
\begin{align}
    F_{T \mid T<H}(t \mid \boldsymbol{x}, z)
    = \frac{\Phi\!\left(\tfrac{\log t - \eta(\boldsymbol{x})}{\sigma}\right)}
           {\Phi\!\left(\tfrac{\log H - \eta(\boldsymbol{x})}{\sigma}\right)},
    \qquad t \in (0,H),
\end{align}
where $\Phi(\cdot)$ denotes the standard normal CDF.

By definition, the non-cured expected event time is
\begin{align}\label{eq:trunc-mean}
    m(\boldsymbol{x}, z)
    = \mathbb{E}[T \mid T<H, \boldsymbol{X}=\boldsymbol{x}, Z=z]
    = e^{\eta(\boldsymbol{x})+\sigma^2/2}\,
      \frac{\Phi\!\left(\tfrac{\log H - \eta(\boldsymbol{x}) - \sigma^2}{\sigma}\right)}
           {\Phi\!\left(\tfrac{\log H - \eta(\boldsymbol{x})}{\sigma}\right)}.
\end{align}

Thus, $m(\boldsymbol{x}, z)$ depends on $\boldsymbol{x}$ only through 
$\eta(\boldsymbol{x})=\lambda_{z0}+\boldsymbol{\lambda}_z^\top \boldsymbol{x}$.
By the chain rule,
\begin{align}
    \nabla_{\boldsymbol{x}} m(\boldsymbol{x}, z)
    = \frac{\partial m(\eta)}{\partial \eta}\,\boldsymbol{\lambda}_z.
\end{align}

\textit{Regularity.} Assume $\boldsymbol{x}$ ranges in a compact set $\mathcal X$
so that $\eta(\boldsymbol{x})\in[\eta_{min},\eta_{max}]$, where $\eta_{\min} = \inf_{\boldsymbol{x}\in\mathcal X}\eta(\boldsymbol{x})$, $\eta_{\max} = \sup_{\boldsymbol{x}\in\mathcal X}\eta(\boldsymbol{x})$.

There exists a finite constant
\begin{align}
C_{\text{time}}
= \sup_{\eta(\boldsymbol{x})\in[\eta_{min},\eta_{max}]}
\left|\frac{\partial m(\eta)}{\partial \eta}\right|
< \infty,
\end{align}

Under this regularity condition,
\begin{align}
    \|\nabla_{\boldsymbol{x}} m(\boldsymbol{x}, z)\|_{W^{-1}}
    = \Big|\tfrac{\partial m(\eta)}{\partial \eta}\Big|\,
      \|\boldsymbol{\lambda}_z\|_{W^{-1}}
    \;\le\; C_{\text{time}} \|\boldsymbol{\lambda}_z\|_{W^{-1}}.
\end{align}

By the dual-norm inequality (mean value theorem combined with Hölder),
\begin{align}
    |m(\boldsymbol{x}, z)-m(\boldsymbol{x}', z)|
    \le \sup_{\boldsymbol{x}} 
    \|\nabla_{\boldsymbol{x}} m(\boldsymbol{x}, z)\|_{W^{-1}}
    \, d_W(\boldsymbol{x}, \boldsymbol{x}'),
\end{align}

which implies the Lipschitz bound
\begin{align}
    L_{\text{time}}(W)
    \le C_{\text{time}} \|\boldsymbol{\lambda}_z\|_{W^{-1}}
    = C_{\text{time}}
      \Big(\sum_{j=1}^p \tfrac{\lambda_{zj}^2}{w_j}\Big)^{1/2}.
\end{align}

Finally, consider minimizing this Lipschitz bound under the equal-scale constraint 
$\sum_{j=1}^p w_j = c > 0$. 
Applying the Cauchy--Schwarz inequality,
\begin{align}
    \Big(\sum_{j=1}^p |\lambda_{zj}|\Big)^2
    = \Big(\sum_{j=1}^p \sqrt{w_j}\cdot \tfrac{|\lambda_{zj}|}{\sqrt{w_j}}\Big)^2
    \le \Big(\sum_{j=1}^p w_j\Big)
       \Big(\sum_{j=1}^p \tfrac{\lambda_{zj}^2}{w_j}\Big)
    = c \sum_{j=1}^p \tfrac{\lambda_{zj}^2}{w_j}.
\end{align}

Equality holds if and only if $w_j \propto |\lambda_{zj}|$. 
Enforcing $\sum_{j=1}^p w_j=c$ yields the unique minimizer
\begin{align}
    w_j^\star = \frac{c\,|\lambda_{zj}|}{\sum_{k=1}^p |\lambda_{zk}|},
    \qquad j=1,\dots,p,
\end{align}
which achieves the minimal Lipschitz bound and gives the optimal
time weight matrix $W_{\text{t},z} = \mathrm{diag}|\boldsymbol{\lambda}_z|$.
\hfill$\square$

\section{ADDITIONAL EXPERIMENTS}

In this section, we present additional simulation studies to further examine the robustness and generalizability of the proposed framework. We begin by outlining the true data-generating mechanisms used across the four experimental settings. We then report supplementary results under different matching group size ($K = 100$) and explore scenarios involving additional confounding variables that affect both treatment assignment and outcomes. Together, these experiments provide a more comprehensive evaluation of the method’s performance under varied data complexities and modeling assumptions.

\subsection{Data-Generating Mechanisms}
\paragraph{Event probability model.}  
For treatment arm $z \in \{0,1\}$, the probability of experiencing the event before the horizon $H$ is modeled by a logistic regression:
\begin{align}
    \Pr(E_z = 1 \mid \bm{X}) 
    &= \mathrm{expit}(\bm{X}^\top \boldsymbol{\beta}_{E_z}),
\end{align}
where $\mathrm{expit}(x) = 1 / (1 + e^{-x})$.  
The corresponding cure probability is $1 - \Pr(E_z = 1 \mid \bm{X})$.  
Treatment effects on the cure probability arise through differences between $\boldsymbol{\beta}_{E_1}$ and $\boldsymbol{\beta}_{E_0}$.  
Accordingly, the ground-truth heterogeneous treatment effect on the cure probability is
\begin{align}
    \mathrm{HTE}_{\text{cure}}(\bm{X}) 
    &= [1 - \Pr(E_1 = 1 \mid \bm{X})] - [1 - \Pr(E_0 = 1 \mid \bm{X})] \nonumber\\
    &= \Pr(E_0 = 1 \mid \bm{X}) - \Pr(E_1 = 1 \mid \bm{X}).
\end{align}

\paragraph{Event time model (for uncured subjects).}  
Conditional on $E_z = 1$, the event time follows a truncated log-normal accelerated failure time (AFT) model:
\begin{align}
    \log(T_z) \mid (E_z = 1, \bm{X}) 
    &\sim \mathcal{N}(\bm{X}^\top \boldsymbol{\lambda}_{z}, \sigma_z^2),
    \quad T_z \le H,
\end{align}
where $H=800$ denotes the administrative censoring horizon.  
The corresponding theoretical mean of the truncated log-normal distribution is
\begin{align}
    \mathbb{E}[T_z \mid \bm{X}, T_z \leq H]
    &= \exp\!\left(\bm{X}^\top \boldsymbol{\lambda}_{z} + \tfrac{1}{2}\sigma_z^2\right)
       \frac{\Phi\!\left( \frac{\log H - \bm{X}^\top \boldsymbol{\lambda}_{z} - \sigma_z^2}{\sigma_z} \right)}
            {\Phi\!\left( \frac{\log H - \bm{X}^\top \boldsymbol{\lambda}_{z}}{\sigma_z} \right)},
\end{align}
where $\Phi(\cdot)$ is the standard normal CDF, and we specified $\sigma_z^2=1$.  
Treatment effects on event timing are induced by differences between $\boldsymbol{\lambda}_{1}$ and $\boldsymbol{\lambda}_{0}$.  
The ground-truth heterogeneous treatment effect on the expected event time is thus defined as
\begin{align}
    \mathrm{HTE}_{\text{time}}(\bm{X})
    &= \mathbb{E}[T_1 \mid \bm{X}, T_1 \leq H] - \mathbb{E}[T_0 \mid \bm{X}, T_0 \leq H].
\end{align}

\paragraph{Observed data generation.}  
For each subject $i = 1, \ldots, n$, the observed survival time $Y_i$ and event indicator $\delta_i$ are generated as
\begin{align}
    T_i &= Z_i T_{1i} + (1 - Z_i) T_{0i}, \\
    E_i &= Z_i E_{1i} + (1 - Z_i) E_{0i}, \\
    C_i &\sim \mathrm{Uniform}(0, 1.5H), \\
    Y_i &= \min(T_i, C_i, H), \quad
    \delta_i = \mathbb{I}(T_i < \min(C_i, H)).
\end{align}
The treatment assignment $Z_i$ is generated according to
\begin{align}
    \Pr(Z_i = 1 \mid \bm{X}_i) = \mathrm{expit}(\bm{X}_i^\top \boldsymbol{\beta}_Z),
\end{align}
where $\boldsymbol{\beta}_Z$ controls the dependence of treatment on covariates.

\subsection{Parameter Settings for the Four Simulation Scenarios}\label{simu: settings}

We consider four simulation scenarios to represent distinct mechanisms of treatment effects on cure probability and event timing. 
Each scenario specifies how the treatment arm modifies the corresponding parameter vectors 
$\boldsymbol{\beta}_{E_z}$ and $\boldsymbol{\Lambda}_{z}$. 

The covariate vector is denoted by 
\[
\bm{X} = (X_1, X_2, \ldots, X_{20})^\top,
\]
where $X_1$–$X_{10}$ are binary variables independently generated from $\mathrm{Bernoulli}(0.5)$, 
and $X_{11}$–$X_{20}$ are continuous variables independently drawn from $\mathcal{N}(0,1)$.

\paragraph{Scenario-specific parameters.}
Across all scenarios, the treatment assignment model is identical:
\begin{align}
    \Pr(Z=1 \mid \bm{X}) &= \mathrm{expit}(0.5 X_{20}),
\end{align}
so that $\boldsymbol{\beta}_Z = (0,\ldots,0,0.5)^\top$.

\medskip
\noindent\textbf{(1) Cure Only.}  
Treatment affects only the event probability, while the event-time model is identical between arms.
\begin{align}
    \Pr(E=1 \mid Z=1, \bm{X}) &= \mathrm{expit}\big(0.7 + 0.5X_{17} + 0.5X_{18} + 0.8X_{19}\big), \\
    \Pr(E=1 \mid Z=0, \bm{X}) &= \mathrm{expit}\big(1.2 + 0.3X_{1} + 0.5X_{2} + 0.2X_{10}\big), \\[4pt]
    \log T \mid (E=1, \bm{X}) &\sim \mathcal{N}\big(3.5 + Z + 0.3X_{2} + 0.5X_{19},\, 1^2\big),
\end{align}
where the shared time model implies $\lambda_{1,0} - \lambda_{0,0} = 1$ (intercept shift only, no heterogeneous effect).

\medskip
\noindent\textbf{(2) Time Only.}  
Treatment affects only the event-time model; the cure probability remains the same across arms.
\begin{align}
    \Pr(E=1 \mid \bm{X}) &= \mathrm{expit}\big(1 + 0.2X_1 - 0.1X_2 - 0.8X_{17} + 0.5X_{18} + 0.4X_{19} - 0.8X_{20}\big), \\[4pt]
    \log T \mid (E=1, Z=1, \bm{X}) &\sim \mathcal{N}\big(4.5 + 0.4X_1 - 0.2X_2 + 0.3X_3 + 0.8X_{15} - 0.6X_{16} - 0.5X_{17},\, 1^2\big), \\
    \log T \mid (E=1, Z=0, \bm{X}) &\sim \mathcal{N}\big(3.5 + 0.2X_1 + 0.2X_2 + 0.5X_3 + 0.2X_{15} - 0.1X_{16} - 0.1X_{17},\, 1^2\big).
\end{align}

\medskip
\noindent\textbf{(3) Both + Independent.}  
Treatment affects both cure and time components, but with disjoint sets of covariates.
\begin{align}
    \Pr(E=1 \mid Z=1, \bm{X}) &= \mathrm{expit}\big(0.7 + 0.3X_{17} + 0.7X_{18} + 0.2X_{19} - 0.8X_{20}\big), \\
    \Pr(E=1 \mid Z=0, \bm{X}) &= \mathrm{expit}\big(1.5\big), \\[4pt]
    \log T \mid (E=1, Z=1, \bm{X}) &\sim \mathcal{N}\big(4.5 + 1.2X_2 - 0.1X_3 + 0.3X_{15},\, 1^2\big), \\
    \log T \mid (E=1, Z=0, \bm{X}) &\sim \mathcal{N}\big(3.5 + 0.3X_2 + 0.2X_3 + 0.3X_{15},\, 1^2\big).
\end{align}

\medskip
\noindent\textbf{(4) Both + Overlap.}  
Treatment affects both mechanisms, and some covariates influence both the cure and time components.
\begin{align}
    \Pr(E=1 \mid Z=1, \bm{X}) &= \mathrm{expit}\big(0.7 + 0.2X_3 + 0.3X_{17} + 0.7X_{18} + 0.2X_{19} - 0.8X_{20}\big), \\
    \Pr(E=1 \mid Z=0, \bm{X}) &= \mathrm{expit}\big(1.5 - 0.4X_3\big), \\[4pt]
    \log T \mid (E=1, Z=1, \bm{X}) &\sim \mathcal{N}\big(4.5 + 1.2X_2 + 0.3X_{15} + 0.5X_{17} - 0.6X_{18},\, 1^2\big), \\
    \log T \mid (E=1, Z=0, \bm{X}) &\sim \mathcal{N}\big(3.5 + 0.3X_2 + 0.4X_3 + 0.3X_{15} + 0.1X_{17} - 0.2X_{18},\, 1^2\big).
\end{align}

\subsection{Results with $K=100$}

To assess the robustness of the proposed method to the choice of the matching parameter $K$, 
we repeated all simulation experiments using $K = 100$ nearest neighbors, 
in addition to the $K = 50$ results presented in the main text. 
Consistent with the main experiments, each simulated dataset contains $n = 20{,}000$ samples, 
with $35\%$ ($7{,}000$) used for training the mixture cure model and $65\%$ ($13{,}000$) used for treatment effect estimation.  
Table~\ref{tab:combined_mae_K100} reports the mean absolute error (MAE, $\pm$SD) 
for heterogeneous treatment effect (HTE) estimation on both cure probability and conditional mean event time (CMET) across the four settings. 

Overall, the results show that the proposed method remains stable with larger matching groups, 
exhibiting only minor variations in estimation accuracy across all scenarios. 
In fact, when increasing $K$ improves the effective match quality without introducing excessive heterogeneity, 
our method tends to achieve slightly higher estimation accuracy.  
As the Cox model does not rely on matching, its results remain unchanged from those in the main text.  
In practice, the optimal choice of $K$ may depend on the available sample size and the trade-off between bias reduction and variance control.

\begin{table*}[htbp]
    \centering
    \caption{MAE ($\pm$SD) of HTE Estimation on Cure Probability ($\times 100$) and CMET Across Four Settings ($K=100$)}
    \resizebox{\textwidth}{!}{
    \begin{tabular}{l|cc|cc|cc|cc}
        \hline
        \textbf{Method} & \multicolumn{2}{c|}{\textbf{Cure Only}} & \multicolumn{2}{c|}{\textbf{Time Only}} & \multicolumn{2}{c|}{\textbf{Both+Indep}} & \multicolumn{2}{c}{\textbf{Both+Overlap}} \\
        & Cure & Time & Cure & Time & Cure & Time & Cure & Time \\
        \hline
        Oracle                & 5.3 $\pm$ 0.4 & 12.3 $\pm$ 0.5 & 4.6 $\pm$ 0.2 & 22.0 $\pm$ 0.8 & 5.2 $\pm$ 0.2 & 16.0 $\pm$ 0.6 & 5.3 $\pm$ 0.2 & 18.2 $\pm$ 0.7 \\
        Partial Oracle        & 5.9 $\pm$ 0.2 & 19.0 $\pm$ 0.7 & 5.1 $\pm$ 0.3 & 27.7 $\pm$ 1.0 & 5.8 $\pm$ 0.3 & 20.8 $\pm$ 0.8 & 6.0 $\pm$ 0.3 & 24.4 $\pm$ 0.9 \\
        \textbf{MCM KNN}      & 6.6 $\pm$ 0.3 & 18.1 $\pm$ 0.6 & 5.4 $\pm$ 0.3 & 27.5 $\pm$ 1.0 & 6.6 $\pm$ 0.3 & 21.2 $\pm$ 0.8 & 6.9 $\pm$ 0.3 & 24.8 $\pm$ 1.0 \\
        MCM KNN (Combined)    & 6.6 $\pm$ 0.3 & 18.3 $\pm$ 0.6 & 5.6 $\pm$ 0.3 & 34.8 $\pm$ 1.2 & 7.0 $\pm$ 0.3 & 24.4 $\pm$ 1.0 & 7.1 $\pm$ 0.4 & 27.5 $\pm$ 1.2 \\
        Feature Selection KNN & 7.0 $\pm$ 0.3 & 19.5 $\pm$ 0.7 & 5.7 $\pm$ 0.3 & 38.9 $\pm$ 1.5 & 8.2 $\pm$ 0.4 & 36.3 $\pm$ 1.3 & 8.4 $\pm$ 0.4 & 36.0 $\pm$ 1.4 \\
        Euclidean KNN         & 8.6 $\pm$ 0.3 & 22.5 $\pm$ 0.8 & 6.0 $\pm$ 0.3 & 44.9 $\pm$ 1.8 & 9.7 $\pm$ 0.5 & 50.1 $\pm$ 1.7 & 10.1 $\pm$ 0.5 & 51.7 $\pm$ 1.9 \\
        Propensity Score KNN  & 16.3 $\pm$ 0.5 & 30.7 $\pm$ 1.0 & 7.6 $\pm$ 0.4 & 90.7 $\pm$ 2.5 & 18.5 $\pm$ 0.6 & 67.2 $\pm$ 2.0 & 18.5 $\pm$ 0.7 & 85.6 $\pm$ 2.8 \\
        Prognostic Score KNN  & 14.3 $\pm$ 0.4 & 29.9 $\pm$ 1.0 & 6.1 $\pm$ 0.3 & 88.0 $\pm$ 2.2 & 12.7 $\pm$ 0.4 & 56.5 $\pm$ 1.8 & 13.6 $\pm$ 0.5 & 67.1 $\pm$ 2.4 \\
        Cox Model (no match)  & 5.2 $\pm$ 0.6 & 36.2 $\pm$ 2.2 & 8.2 $\pm$ 0.7 & 78.7 $\pm$ 3.1 & 7.7 $\pm$ 0.7 & 58.0 $\pm$ 1.3 & 8.6 $\pm$ 0.7 & 67.5 $\pm$ 1.5 \\
        \hline
    \end{tabular}
    }
    \label{tab:combined_mae_K100}
\end{table*}

\subsection{Additional Experiments under Mild Confounding}

To evaluate robustness to treatment–outcome confounding, 
we repeated all simulations using the same cure and event time models as in Section~\ref{simu: settings}, 
but modified the treatment assignment mechanism to depend on additional outcome-related covariates. 
This induces mild confounding, as some predictors jointly affect treatment and the outcomes. 
Specifically, for each of the four simulation settings, the treatment assignment model is given by

\medskip
\noindent\textbf{(1) Cure Only.}
\begin{align}
    \Pr(Z=1 \mid \bm{X}) 
    &= \mathrm{expit}(0.2 X_{2} + 0.8 X_{3} + 0.3 X_{19} - 0.5 X_{20} + 0.5).
\end{align}

\noindent\textbf{(2) Time Only.}
\begin{align}
    \Pr(Z=1 \mid \bm{X}) 
    &= \mathrm{expit}(0.3 X_{2} - 0.2 X_{3} + 0.4 X_{16} + 0.5 X_{17} + 0.5).
\end{align}

\noindent\textbf{(3) Both + Independent.}
\begin{align}
    \Pr(Z=1 \mid \bm{X}) 
    &= \mathrm{expit}(0.3 X_{3} - 0.2 X_{4} + 0.4 X_{16} + 0.6 X_{17} + 0.5).
\end{align}

\noindent\textbf{(4) Both + Overlap.}
\begin{align}
    \Pr(Z=1 \mid \bm{X}) 
    &= \mathrm{expit}(0.4 X_{3} + 0.3 X_{16} + 0.6 X_{17} - 0.5 X_{18} + 0.5).
\end{align}

\begin{table*}[htbp]
    \centering
    \caption{MAE ($\pm$SD) of HTE Estimation under Mild Confounding ($K=50$)}
    \resizebox{\textwidth}{!}{
    \begin{tabular}{l|cc|cc|cc|cc}
        \hline
        \textbf{Method} & \multicolumn{2}{c|}{\textbf{Cure Only}} & \multicolumn{2}{c|}{\textbf{Time Only}} & \multicolumn{2}{c|}{\textbf{Both+Indep}} & \multicolumn{2}{c}{\textbf{Both+Overlap}} \\
        & Cure & Time & Cure & Time & Cure & Time & Cure & Time \\
        \hline
        Oracle                & 6.6 $\pm$ 0.2 & 17.4 $\pm$ 1.3 & 6.3 $\pm$ 0.2 & 25.6 $\pm$ 1.6 & 6.6 $\pm$ 0.3 & 20.2 $\pm$ 1.2 & 6.7 $\pm$ 0.3 & 22.5 $\pm$ 1.4 \\
        Partial Oracle        & 7.4 $\pm$ 0.4 & 28.3 $\pm$ 1.5 & 7.4 $\pm$ 0.3 & 35.0 $\pm$ 1.9 & 7.7 $\pm$ 0.4 & 29.7 $\pm$ 1.6 & 7.8 $\pm$ 0.4 & 32.5 $\pm$ 1.7 \\
        \textbf{MCM KNN}      & 8.0 $\pm$ 0.4 & 26.8 $\pm$ 1.2 & 7.8 $\pm$ 0.4 & 40.6 $\pm$ 2.3 & 8.2 $\pm$ 0.5 & 33.4 $\pm$ 1.9 & 8.1 $\pm$ 0.5 & 35.5 $\pm$ 1.9 \\
        MCM KNN (Combined)    & 8.2 $\pm$ 0.5 & 27.2 $\pm$ 1.3 & 7.7 $\pm$ 0.4 & 34.7 $\pm$ 2.0 & 8.2 $\pm$ 0.5 & 30.5 $\pm$ 1.7 & 8.6 $\pm$ 0.4 & 33.4 $\pm$ 1.8 \\
        Feature Selection KNN & 8.3 $\pm$ 0.3 & 27.6 $\pm$ 1.4 & 7.8 $\pm$ 0.5 & 40.9 $\pm$ 2.2 & 8.5 $\pm$ 0.5 & 33.3 $\pm$ 1.8 & 8.6 $\pm$ 0.5 & 36.1 $\pm$ 1.9 \\
        Euclidean KNN         & 9.6 $\pm$ 0.6 & 30.6 $\pm$ 1.7 & 8.4 $\pm$ 0.6 & 48.9 $\pm$ 2.5 & 10.5 $\pm$ 0.7 & 53.1 $\pm$ 2.4 & 11.0 $\pm$ 0.8 & 57.2 $\pm$ 2.5 \\
        Propensity Score KNN  & 16.5 $\pm$ 0.8 & 34.4 $\pm$ 1.6 & 9.6 $\pm$ 0.5 & 81.1 $\pm$ 2.7 & 19.1 $\pm$ 0.8 & 70.4 $\pm$ 2.6 & 18.8 $\pm$ 0.8 & 66.8 $\pm$ 2.3 \\
        Prognostic Score KNN  & 14.0 $\pm$ 0.7 & 36.2 $\pm$ 1.8 & 7.7 $\pm$ 0.5 & 91.1 $\pm$ 4.1 & 15.7 $\pm$ 0.7 & 54.9 $\pm$ 2.3 & 16.2 $\pm$ 0.8 & 63.2 $\pm$ 2.4 \\
        Cox Model (no match)  & 4.0 $\pm$ 0.4 & 39.4 $\pm$ 2.3 & 9.7 $\pm$ 0.6 & 74.4 $\pm$ 3.5 & 7.4 $\pm$ 0.5 & 57.7 $\pm$ 2.6 & 8.4 $\pm$ 0.5 & 69.1 $\pm$ 1.7 \\
        \hline
    \end{tabular}
    }
    \label{tab:confounded_mae}
\end{table*}

\paragraph{Results.}
Table~\ref{tab:confounded_mae} reports the mean absolute error (MAE, $\pm$SD) 
for heterogeneous treatment effect (HTE) estimation on cure probability and conditional mean event time (CMET) 
under these mildly confounded treatment models, using $K = 50$ nearest neighbors.

Compared with the non-confounded baseline, the introduction of mild confounding leads to a modest increase in MAE, as expected, yet the overall performance ranking remains unchanged. The proposed MCM-based estimators maintain estimation accuracy close to the partial oracle across all four settings, indicating strong robustness of the learned distance metrics. These results suggest that the method can tolerate moderate dependence between treatment assignment and covariates without substantial degradation in HTE estimation.

\end{document}